\documentclass[journal]{IEEEtran}

\usepackage{cite}
\usepackage{algorithm,algpseudocode}

\usepackage[dvips]{graphicx}
\usepackage[font=small]{caption}
\usepackage{subfig}
\usepackage{float}

\usepackage{amssymb}
\usepackage{multirow}
\usepackage{colortbl}
\usepackage{array, tabularx}
\usepackage{color}
\usepackage{amsthm}
\usepackage{amsmath}
\usepackage{setspace}
\usepackage{lettrine}
\usepackage{ctable}

\begin{document}
	
%\begin{titlepage}
%\clearpage

%\vspace*{1cm}
%\onehalfspacing
\title{
\onehalfspacing
On Reusing Pilots Among Interfering Cells in Massive MIMO
}

\author{Jy-yong Sohn, Sung Whan Yoon,~\IEEEmembership{Student Member,~IEEE,}
	and~Jaekyun~Moon,~\IEEEmembership{Fellow,~IEEE}% <-this % stops a space
%\thanks{M. Shell is with the Department of Electrical and Computer Engineering, Georgia Institute of Technology, Atlanta, GA, 30332 USA e-mail: (see http://www.michaelshell.org/contact.html).}% <-this % stops a space
\thanks{The authors are with the School
of Electrical Engineering, Korea Advanced Institute of Science and Technology, Daejeon, 305-701, Republic of Korea (e-mail: jysohn1108@kaist.ac.kr, shyoon8@kaist.ac.kr, jmoon@kaist.edu). This paper was presented in part at IEEE ICC 2015 Workshop on 5G \& Beyond \cite{sohn2015pilots}.}% <-this % stops a space
%\thanks{Manuscript received April 19, 2005; revised December 27, 2012.}
}


\markboth{submitted to IEEE Transactions on Wireless Communications}%
{Shell \MakeLowercase{\textit{et al.}}: Bare Demo of IEEEtran.cls for Journals}



\maketitle
%\thispagestyle{empty}
\begin{abstract}
	Pilot contamination, caused by the reuse of pilots among interfering cells, remains as a significant
	obstacle that limits the performance of massive multi-input multi-output antenna systems. To
	handle this problem, less aggressive reuse of pilots involving allocation of additional pilots
	for interfering users is closely examined in this paper. Hierarchical pilot reuse methods are proposed,
	which effectively mitigate pilot contamination and increase the net throughput of the system. Among
	the suggested hierarchical pilot reuse schemes, the optimal way of assigning pilots to different users
	is obtained in a closed-form solution which maximizes the net sum-rate in a given coherence time.
	Simulation results confirm that when the ratio of the
	channel coherence time to the number of users in each cell is sufficiently large, less aggressive reuse of
	pilots yields significant performance advantage relative to the case where all cells reuse the same pilot
	set.
%\textcolor{red}{Pilot contamination, caused by the reuse of pilots among interfering cells, is remained as a significant problem which limits the performance of massive multi-input multi-output (MIMO) antenna systems.
%	To handle this problem, the effect of less aggressive reuse of pilots involving allocation of additional pilots for interfering users is closely examined in this paper.	
%	%In this paper, the effect of pilot reuse on the pilot contamination problem is closely examined, considering scenarios of using pilot sequences possibly longer than the number of users in a cell.
%Hierarchical pilot reuse methods are proposed, which effectively mitigate the pilot contamination and increase the net throughput of the system.
%Among the suggested hierarchical pilot reuse schemes, the optimal way of assigning pilots to different users is obtained in a closed-form solution which maximizes the net sum-rate in a given coherence time.
%The optimal solution is based on mathematical proofs, and the result provides a physical insight on the effective way of reducing pilot contamination.
%Simulation results confirm that when the ratio of the channel coherence time to the number of users in each cell is sufficiently large, less aggressive reuse of pilots yields significant performance advantage relative to the case where all cells reuse the same pilot set.}
\end{abstract}

\begin{IEEEkeywords}
Massive MIMO, Multi-user MIMO, Multi-cell MIMO, Pilot contamination, Pilot assignment, Pilot reuse, Interference, Large-scale antenna system, Channel estimation
\end{IEEEkeywords}

%\end{titlepage}

% \fi
%\IEEEpeerreviewmaketitle

%\newpage
%\setcounter{page}{1}
\markboth{submitted to IEEE Transactions on Wireless Communications}%
{Shell \MakeLowercase{\textit{et al.}}: Bare Demo of IEEEtran.cls for Journals}

\section{Introduction}
%% MIMO, Massive MIMO (only pilot contamination remain)

\lettrine{S}{upporting} the stringent requirements of the next generation wireless communication systems is an ongoing challenge, especially given the foreseeable scenarios where massively-deployed devices rely on applications with high-throughput (e.g. virtual/augmented reality) and/or low latency (e.g. autonomous vehicles). This unavoidable trend gives rise to discussions on the technology for next generation communications, collectively referred to as 5G. 
Regarding the engineering requirements of high data rate, low latency, and high energy efficiency, attractive 5G technologies in \cite{andrews2014will, boccardi2014five, Osseiran2014, Wang2014} include massive multi-input multi-output (MIMO) \cite{rusek2013scaling,larsson2014massive, lu2014overview}, ultra-densification \cite{gotsis2016ultradense, bhushan2014network,ge20165g} and millimeter wave (mm-Wave) communications among other things.

Massive MIMO, or a large-scale antenna array system, is the deployment of a very large number of antenna elements at base stations (BSs), possibly orders of magnitude larger than the number of user terminals (UTs) served by each BS \cite{marzetta2006much, marzetta2010noncooperative}. Assuming $M$, the number of BS antennas, increases without bound, the asymptotic analysis on capacity and other fundamental aspects of massive MIMO were presented in \cite{marzetta2010noncooperative, ngo2013energy}. 
The work of \cite{marzetta2010noncooperative}, in particular, has demonstrated that in time-division duplex (TDD) operation with uplink training for attaining channel-state-information (CSI), the effects of uncorrelated noise and fast fading disappear as $M$ grows to infinity, with no cooperation necessary among BSs.
%; only the effect of degraded channel estimation due to pilot contamination from the reuse of the same pilot set in interfering cells limits the sum-rate performance. 
According to this ``channel-hardening" behavior, \cite{marzetta2010noncooperative} concluded that 
in a single-cell setting where UTs have orthogonal pilots, capacity increases linearly with $K$, the number of UTs, as $M/K$ increases even without channel knowledge. However, in multi-cell setting, the channel estimation errors due to pilot contamination prevent the linear growth of the sum rate with $K$.
Pilot contamination is caused by the reuse of pilots among different cells
and persists even as the number of BS antennas increases without limit; it remains as a fundamental issue in massive MIMO \cite{elijah2015comprehensive}.

%% Previous research on mitigating pilot contamination problem
Various researchers have since investigated systematic methods to mitigate the pilot contamination effect 
\cite{yin2013decontaminating,yin2013coordinated,jose2011pilot,ashikhmin2012pilot,appaiah2010pilot,fernandes2013inter,vu2014successive,lee2014mitigation, zhu2015smart, nguyen2015resource, li2015multi, bjornson2016massive, saxena2015mitigating, sohn2015pilots}.
Some researchers exploited the angle-of-arrival (AoA) information to combat pilot contamination.
A coordinated pilot assignment strategy was suggested based on AoA information which can eliminate channel estimation error as the number of BS antennas increases without bound \cite{yin2013decontaminating}, \cite{yin2013coordinated}. However, favorable AoA distributions are not always guaranteed in real scenarios. In addition, significant challenges remain on actually getting AoA information and managing network overload for cooperation.

Another approaches have focused on precoding to reduce the pilot contamination effect. 
The work of \cite{jose2011pilot} suggested a distributed single-cell linear minimum mean-square-error (MMSE) precoding, which exploits pilot sequence information to 
minimize the error caused by inter-cell and intra-cell interference. The proposed precoding scheme has improved performance compared to the conventional single-cell zero-forcing (ZF) precoder, but cannot completely eliminate pilot contamination. In \cite{ashikhmin2012pilot}, outer multi-cellular precoding under the assumption of cooperating BSs was introduced for mitigating the pilot contamination effect. The suggested method is based on adjusting the precoding vector according to the contaminated channel estimate via coordinated information. Backhaul network overload for realizing cooperation remains as an issue in realizing this precoding scheme. 

The impact of pilot transmission protocol is considered in \cite{appaiah2010pilot,fernandes2013inter,vu2014successive,lee2014mitigation}. Shifting of pilot frames corresponding to neighboring cells was proposed, which mitigated pilot contamination \cite{appaiah2010pilot}. Appropriate power allocation to increase the signal-to-interference ratio (SIR) in shifted pilot frame was also suggested \cite{fernandes2013inter}. Even though the shifted pilot frame method avoided correlation between identical pilot sequences, it caused the correlation between pilot sequence and data sequence in the training phase. Central control is also required to achieve precise timing among shifted pilot frames, which needs additional overload. Another pilot transmission protocol to eliminate pilot contamination was proposed in \cite{vu2014successive} based on imposing a silent phase for each cell
over successive time slots during the training period. A linear combination of 
observations taken over the successive training phases allow simple elimination 
of inter-cell interferences while fully reusing pilots across cells, but the method does not offer any advantage in terms of reducing total training time overhead
compare to the straightforward employment of $LK$ orthogonal pilot symbols over all $LK$ users across entire cells.

Pilot assignment strategies are also considered as an alternative to mitigate pilot contamination and to increase net throughput \cite{lee2014mitigation, zhu2015smart, nguyen2015resource, ngo2016cell}. 
According to the strategy suggested in \cite{lee2014mitigation}, the cell is divided into the center area and the edge area; the users in the center employ the same pilot resource non-cooperatively, while cooperative resource allocation is applied to those in the edge. This strategy mitigates severe pilot contamination for edge users, but still requires cooperation among BSs.
A greedy algorithm (called smart pilot assignment) is suggested in \cite{zhu2015smart}, which assigns pilot with less inter-cell interference to the user with low channel quality. This approach improved the minimum signal-to-interference-plus-noise ratio (SINR) within each cell compared to random assignment case, but requires cooperation between BSs in order to exploit slow fading coefficients of entire cells.
In \cite{nguyen2015resource}, spectral efficiency (SE) was maximized with respect to pilot assignment, power allocation and the number of antennas.
Considering cell-free massive MIMO systems where $M$ antennas and $K$ users are dispersed in a region, the authors of \cite{ngo2016cell} jointly optimized greedy pilot assignment and power control to increase the achievable rate. However, the greedy pilot assignment requires backhaul network, involving cooperation between distributed antennas.
According to these papers, optimal pilot assignment can increase SE compared to random assignment, while finding the optimal solution also require BS cooperation. Moreover, \cite{ngo2016cell} considers pilot assignment and power control for cell-free distributed antenna systems, whereas the present paper considers pilot allocation in multi-cellular systems with co-located antenna setting.%These research gives us message of \textit{appropriate pilot assignment can effectively reduce pilot contamination}.

As seen above, most of the known effective solutions to combat pilot contamination tend to rely heavily on cooperation among BSs, leading to backhaul overload issues.
Also, most of the previous works 
%dealing with improved channel estimation and precoding methods 
assume full reuse of the same pilot set among all cells
and rarely discuss the 
potential associated with allowing more orthogonal pilots.

Some researchers investigated the effect of employing less aggressive pilot reuse methods, but to limited extents  \cite{li2015multi, bjornson2016massive,  saxena2015mitigating}. A scenario of utilizing the number of pilots greater than $K$ is considered in \cite{li2015multi}, which suggested a multi-cell MMSE detector exploiting all pilots in the system.  
It is shown that the suggested multi-cell MMSE scheme can significantly increase spectral efficiency as the pilot reuse factor becomes less aggressive.
An SE-maximizing massive MIMO system is considered in \cite{bjornson2016massive}, which optimizes the number of scheduled users $K$ for given $M$ and pilot reuse factor. The simulation result shows that the optimal solution selects less aggressive pilot reuse for some practical scenarios.
A pilot reuse factor of 3 (cell partitioning similar to frequency reuse of 3) is considered in \cite{saxena2015mitigating}, which is shown to be beneficial to mitigate pilot contamination. 
However, only symmetric pilot reuse patterns (lattice structure) was considered in \cite{saxena2015mitigating}, and no
closed-form solution was found for the optimal pilot reuse
factor; as such, the impact of partial pilot reuse was not made clear.
%However, they considered symmetric pilot reuse patterns (lattice structure) only, and the optimal pilot reuse factor solution is not obtained in a closed-form manner so that the impact of less aggressive pilot reuse is not shown explicitly.

In contrast, in \cite{sohn2015pilots} the present authors presented a systematically-constructed pilot reuse method to effectively reduce
pilot contamination. Non-trivial hierarchical pilot reuse/assignment schemes are proposed, and a closed-form solution to optimum pilot assignment 
is presented that maximizes the net sum-rate. The
present paper further analyzes the pilot assignment strategy proposed in \cite{sohn2015pilots}, providing formal
proofs of the previously presented mathematical results and offering key insights into
the physical significance of the optimal pilot assignment rule, along the way. 
Moreover, the present paper adds the analysis and finds optimal pilot assignment associated with a large but finite number of antennas,
whereas \cite{sohn2015pilots} only provided an asymptotic result for BS with an infinite number of antennas. 
Physical insights on the optimal assignment for infinite $M$ can be similarly observed for the finite $M$ case, while numerical results show that the net throughput between optimal assignment and conventional full pilot reuse still has a substantial gap, even in practical scenarios of deploying $100-1000$ BS antennas.
 The present paper also investigates the ideal portion of pilot training allocated for the optimal assignment strategy; this investigation points to the interesting fact that in the optimal scheme a non-vanishing portion of the coherence interval is reserved for the pilots as $N_{coh}/K$ grows. 
 
%\textcolor{red}{In contrast, probably the work of \cite{sohn2015pilots} by the present authors is the first attempt which
%	deliberately selected the pilot reuse scheme to effectively reduce the pilot contamination. 
%	Non-trivial hierarchical pilot reuse/assignment schemes are proposed,
%	 and the net sum-rate maximizing pilot assignment method is obtained in a closed-form solution. The present paper further analyzes the pilot assignment strategy proposed in \cite{sohn2015pilots}; provides formal proofs of mathematical results for the first time, shows the impact of less aggressive pilot reuse by finding the physical meaning of the optimal pilot assignment rule, investigates the portion of pilot training allocated for the optimal assignment strategy, and suggests practical scenarios where the proposed strategy shows outstanding performance.}
 
Overall, this paper is about finding the optimal portion of pilot transmission as well 
as pilot reuse strategy which maximize the net sum-throughput, taking into account pilot contamination due to interfering cells. %and without imposing a fixed $N_{coh}/K$ ratio.
This problem arises in some meaningful practical scenarios but has not been addressed previously. 
Although the optimal portion of pilot training was considered by several researchers \cite{ngo2014massive, hassibi2003much, marzetta2006much}, none of their works or any other previous related works to our best knowledge addressed the optimal pilot length and optimal pilot assignment in massive MIMO, where pilot contamination limits the performance and no constraint is imposed on the specific pilot reuse method.

Compared to the ideal asymptotic analysis in \cite{marzetta2010noncooperative}, several works attempted to emphasize the practical aspects of applying massive MIMO in real-world systems. Considering BSs with a finite $M$, the achievable rate is obtained \cite{hoydis2013massive} when matched-filter (MF) or MMSE detection is used. Focusing on some propagation environments where the channel hardening does not hold, \cite{ngo2016no} suggested a downlink blind channel estimation method. The behavior of massive MIMO systems with non-ideal hardware is observed in \cite{bjornson2014massive}.

%% Organize
This paper is organized as follows. Section II gives an overview of the system model for massive MIMO and the 
pilot contamination effect.
Section III discusses the assumptions made 
on cell geometry and basic partitioning steps needed to utilize
longer pilots, and presents the mathematical analysis for specifying the optimal pilot assignment strategy.
The mathematical
results are first stated in \cite{sohn2015pilots}, while the formal proofs and the key insights leading to physical interpretations of the obtained
results are provided in the present paper.
%\textcolor{red}{The mathematical results are stated in \cite{sohn2015pilots}, while the formal proofs and the detailed physical meanings of obtained results are provided for the first time in the present paper.}
In Section IV, simulation results on the performance of the optimal pilot assignment are presented.
Section V includes further comments on finite $M$ analysis, optimal portion of pilots, cell partitioning in conventional frequency reuse, consideration of prioritized users and application in ultra-dense networks.  
Section VI finally draws conclusions.

\section{Pilot Contamination Effect for Massive MIMO Multi-Cellular System}

%% System Model
\subsection{System Model}\label{Section:SystemModel}
We assume that the network consists of $L$ hexagonal cells with $K$ users per cell who are uniform-randomly located.
Downlink CSIs are estimated at each BS by uplink pilot training assuming channel reciprocity in TDD operation. 
This paper assumes the channel model in \cite{bjornson2016massive}; this model is appropriate for sufficiently large ($M>32$) numbers of BS antennas, as tested in 
\cite{gao2015massive}.
The singular value spread of the measured channel is close to that of an i.i.d. Rayleigh channel model for a large number of antennas \cite{gao2015massive}. Therefore, we adopt this model.
The complex propagation coefficient $g$ of a link can be decomposed into a complex fast fading factor $h$ and a slow fading factor $\beta$. Therefore, the channel between the $m^{th}$ BS antenna of the $j^{th}$ cell and the $k^{th}$ user of the $l^{th}$ cell is modeled as
%\begin{equation*} \label{channel model}
$g_{mjkl}=h_{mjkl}\sqrt{\beta_{jkl}}.$
%\end{equation*}
The slow fading factor, which accounts for the geometric attenuation and shadow fading, is modeled as
%\begin{equation*} \label{slow fading model}
$\beta_{jkl}=(\frac{1}{r_{jkl}})^\gamma$
%\end{equation*}
where $r_{jkl}$ is the distance between the $k^{th}$ user in the $l^{th}$ cell and the base station in the $j^{th}$ cell. The parameter $\gamma$ represents the signal decay exponent.
Let $T_{coh}$ be the coherence time interval and $T_{del}$ be the channel delay spread. It is convenient to express the coherence time interval as a dimensionless quantity, $N_{coh}=T_{coh}/T_{del}$, via normalization by channel delay spread.

%% Pilot Contamination Effect
\subsection{Pilot Contamination Effect}
The pilot contamination effect is the most serious issue that arises in multi-cell TDD systems with very large BS antenna arrays. 
For uplink training, each BS collects pilot sequences sent by its users. Usually, orthogonal pilot sequences are assigned to users in a cell so that the channel estimate for each user 
does not suffer from interferences from other users in the same cell. However, the use of the same pilot sequences for users in other cells cause the channel estimates to be contaminated, 
and this effect, called pilot contamination, limits the achievable rate, even as $M$, the number of BS antennas, tends to infinity. 

According to \cite{marzetta2010noncooperative}, under the assumption of a single user per cell, the achievable rate during uplink data transmission for the user in the $j^{th}$ cell contaminated by users with the same pilot on other cells is given for a large $M$ by
\begin{equation} \label{achievableR}
\log_{2} \left(1+\frac{\beta_{jj}^{2}}{\sum_{l\neq j}\beta_{jl}^{2}}\right)
\end{equation}
where $\beta_{jl}$ is the slow fading component of the channel between the  $j^{th}$ BS and the interfering user in the $l^{th}$ cell. In the limit of large $M$, the achievable rate depends only on the ratio of the signal to interference due to the pilot reuse.

% Pilot Assigning Strategy
\section{Pilot Assignment Strategy}
In this section, we provide analysis on how much time should be allocated for channel training given a coherence time and how the pilot sequences should be assigned to users on multiple cells.
We derive optimal pilot assignment strategy, which mitigates the pilot contamination effect and maximizes total bits transmitted in a given $T_{coh}$. 
Our analysis considers using pilot sequences possibly longer than the number of users in each cell, while orthogonality of the pilots within a cell is guaranteed. 
%For simplicity, all cells are restricted to have only one user, i.e, $K=1$ in this section. The results are extended to $K>1$ in the next section.

%% Hexagonal-Lattice Based Cell Clustering
\subsection{Hexagonal-Lattice Based Cell Clustering and Pilot Assignment Rule}

%\begin{figure}[!t]
%    \centering
%    \includegraphics[height=25mm]{3_way_partitioning.pdf}
%    \caption{3-way Partitioning}
%    \label{Fig:3-lattice Partitioning}
%\end{figure}

Consider $L$ hexagonal cells. Imagine partitioning these cells into three equi-distance subsets 
maintaining the same lattice structure as depicted in Fig. \ref{Fig:3Lattice}. 
This partitioning is identical to the familiar partitioning of contiguous hexagonal cells for utilizing three frequency bands according to a frequency reuse factor of three.

%\begin{figure}[!t]
%\centering
%        \includegraphics[height=30mm]{Hierarchy.pdf}
%    \caption{Hierarchical set partitioning}
%    \label{Fig:Hierarchy tree structure}
%\end{figure}

\begin{figure}
	\centering
	\subfloat[][3-way Partitioning]{\includegraphics[width=25mm]{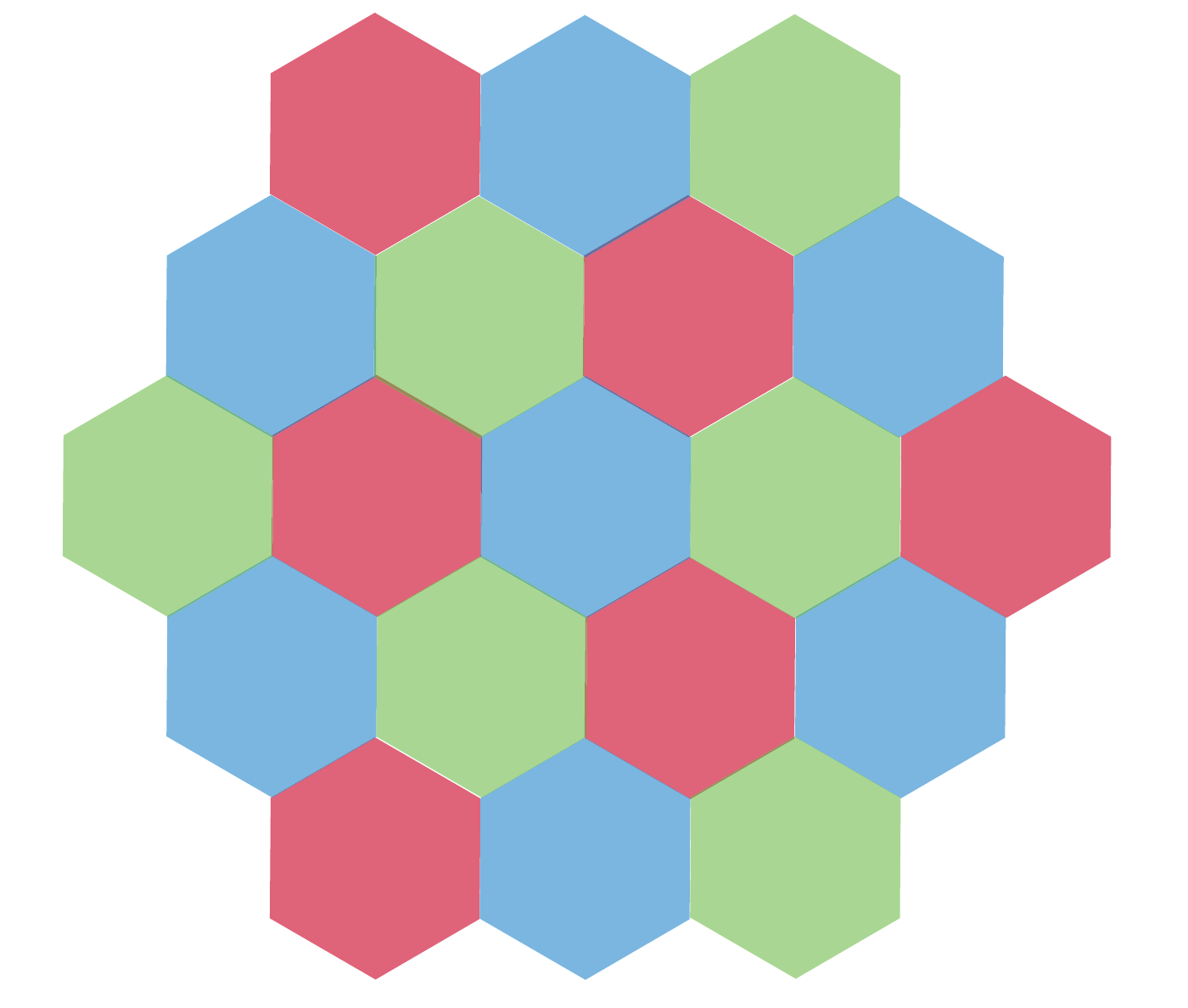}\label{Fig:3Lattice}}
	\quad \quad
	\subfloat[][Hierarchical set partitinoing]{\includegraphics[width=0.45\textwidth]{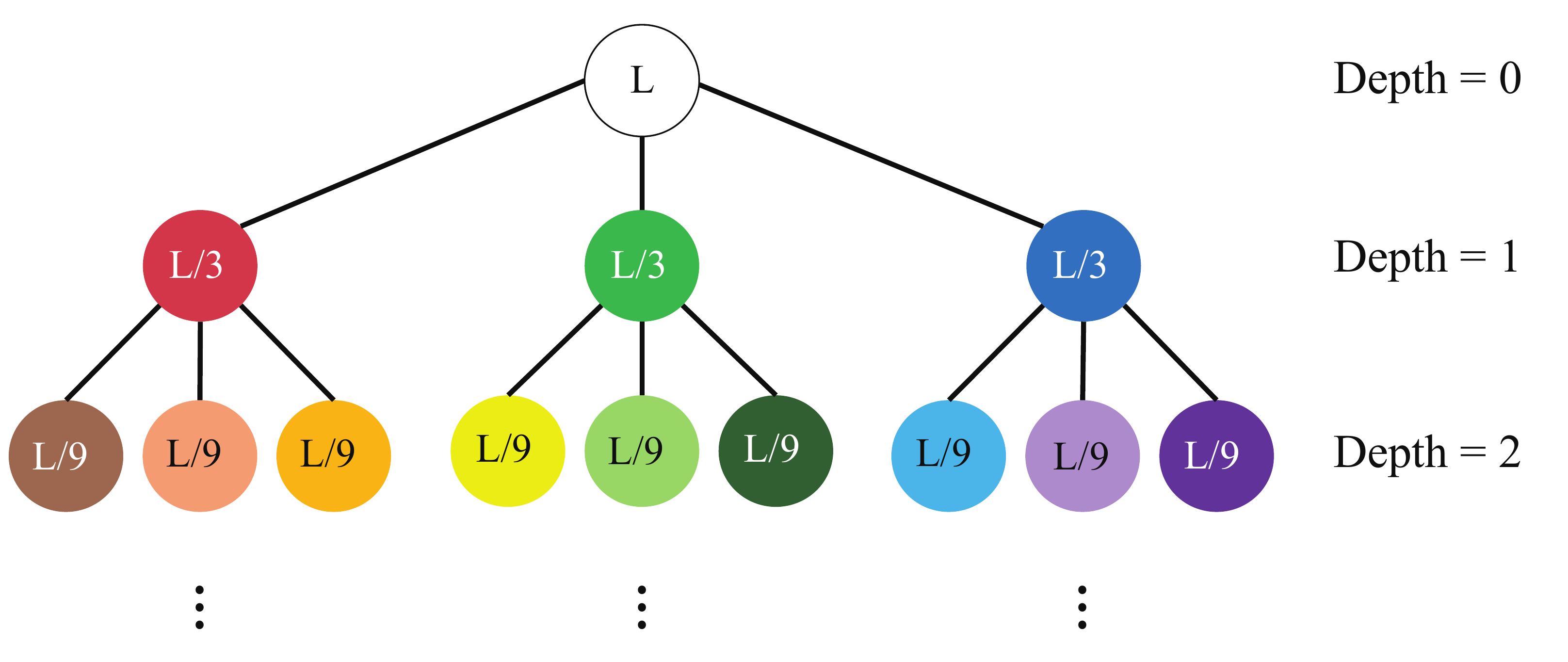}\label{Fig:Hierarchy}}
	\caption{The Cell Partitioning Method}
	\label{steady_state}
\end{figure}

It can easily be seen that each coset, having the same hexagonal lattice structure, can be further partitioned in the similar way.
The partitioning can clearly be applied in a successive fashion, giving rise to the possibility of hierarchical set partitioning of the entire cells.
In the tree structure of Fig. \ref{Fig:Hierarchy}, the root node (at depth $0$) represents the original group of  $L$ contiguous hexagonal cells,
and the three child nodes labeled $L/3$ correspond to the three colored-cosets of Fig. \ref{Fig:3Lattice}. Also, applying a 3-way partitioning to a coset results in additional three child nodes with labels $L/9$. Note that a node at depth $i$ corresponds to a subset of $L3^{-i}$ cells. 

This cell clustering method is used to define pilot assignment rule in multi-cell system. Consider when each cell has a single user (i.e., $K=1$). Then, determining pilot sharing cells specifies the pilot assignment rule. In our approach, cells with same color are defined to share same pilot, while cells with different color have orthogonal pilots. In the corresponding tree structure, the number of leaves with different (non-white) colors represents the number of different orthogonal pilot sets used.
A leaf (end node) with a single cell would correspond to a depth of $\log_{3}L$, but we do not allow such leaves in our analysis as this means
there would be users with no pilot contamination, thus driving the average per-cell throughput of the network to infinity 
as $M$ grows. This particular situation would not lend itself to a meaningful mathematical analysis. Thus, the maximum depth of a leaf in our tree is set to $\log_3 L -1$. 

The pilot assignment rule can be generalized to the case of having multiple users per cell (i.e., $K>1$), by applying the procedure $K$ times consecutively ($K$ users in each cell obtain orthogonal pilots by $K$ independent procedures). Note that $K$ procedures result in $K$ partitioning trees, where each tree represents pilot assigning rule of $L$ users.

Note that the suggested pilot assignment strategy assumes uniform power allocation: every user has the same transmit power for pilot. The suggested assignment focuses on reducing pilot contamination by placing pilot-sharing users in distant cells, and the optimal power allocation has not been considered here. According to the recent research \cite{ngo2016cell} on pilot assignment and power control \textit{within a cell}, controlling transmit power for pilot among different users has a significant role in increasing spectral efficiency. This could be an interesting subject that could be pursued in conjunction with the suggested pilot assignment strategy, but we leave the pilot power control issue as a future research topic.  

%% Pilot Assignment Vector
\subsection{Pilot Assignment Vector}\label{Section:PilotAssignmentVector}

The pilot assignment method can conveniently be formulated in a vector form. Let $\mathbf{p}$ be a vector with element 
$p_{i}$, $i=0,\cdots,\log_3 L -1$, representing the number of leaves at depth $i$ of the partitioning tree. For example, 
for the tree of Fig. \ref{Fig:coloring example}, we have $\mathbf{p}=(0,2,3,0)$, as there are two leaves at depth 1 and three at depth 2.

\newtheorem{theorem}{Theorem}
\newtheorem{lemma}{Lemma}
\newtheorem{corollary}{Corollary}
\newenvironment{definition}[1][Definition]{\begin{trivlist}
\item[\hskip \labelsep \normalfont #1]}{\end{trivlist}}

\begin{definition}[Definition:\nopunct]
For the given $L$ cells, denote $P_{L,K}$ as the set of valid pilot assignment vectors based on 3-way partitioning, which is given by
\begin{align} \label{pilot assgning vector}
P_{L,K} = \{\mathbf{p}&=(p_0,p_1,\cdots,p_{\log_3 L-1})\: : \nonumber\\
&0 \leq p_i \leq K3^i, \sum\limits_{i=0}^{\log_3 L -1} p_i3^{-i}=K\}
\end{align}
\end{definition}
where it is implied that $\ p_{i}$ is an integer.
\begin{definition}[Definition:\nopunct]
For the given $L$ cells, $N_{pil}(\mathbf{p})$ is the length of the valid pilot assignment vector $\mathbf{p}=(p_0,p_1,\cdots,p_{\log_3 L -1})$ and is given by
%\begin{eqnarray*} \label{pilot length function}
$N_{pil}(\mathbf{p}) = \sum\limits_{i=0}^{\log_3 L-1} p_i.$
%\end{eqnarray*}
\end{definition}
The pilot length $N_{pil}(\mathbf{p})$ represents the number of orthogonal pilots corresponding to the given pilot assignment vector $\mathbf{p}$.
In the corresponding tree structure, $N_{pil}(\mathbf{p})$ is the number of leaves.
As an example, for the pilot assignment strategy shown in Fig. \ref{Fig:coloring example}, we have $N_{pil}(\mathbf{p})=5$.

\begin{figure}[!t]
	
	\centering
	\includegraphics[height=30mm]{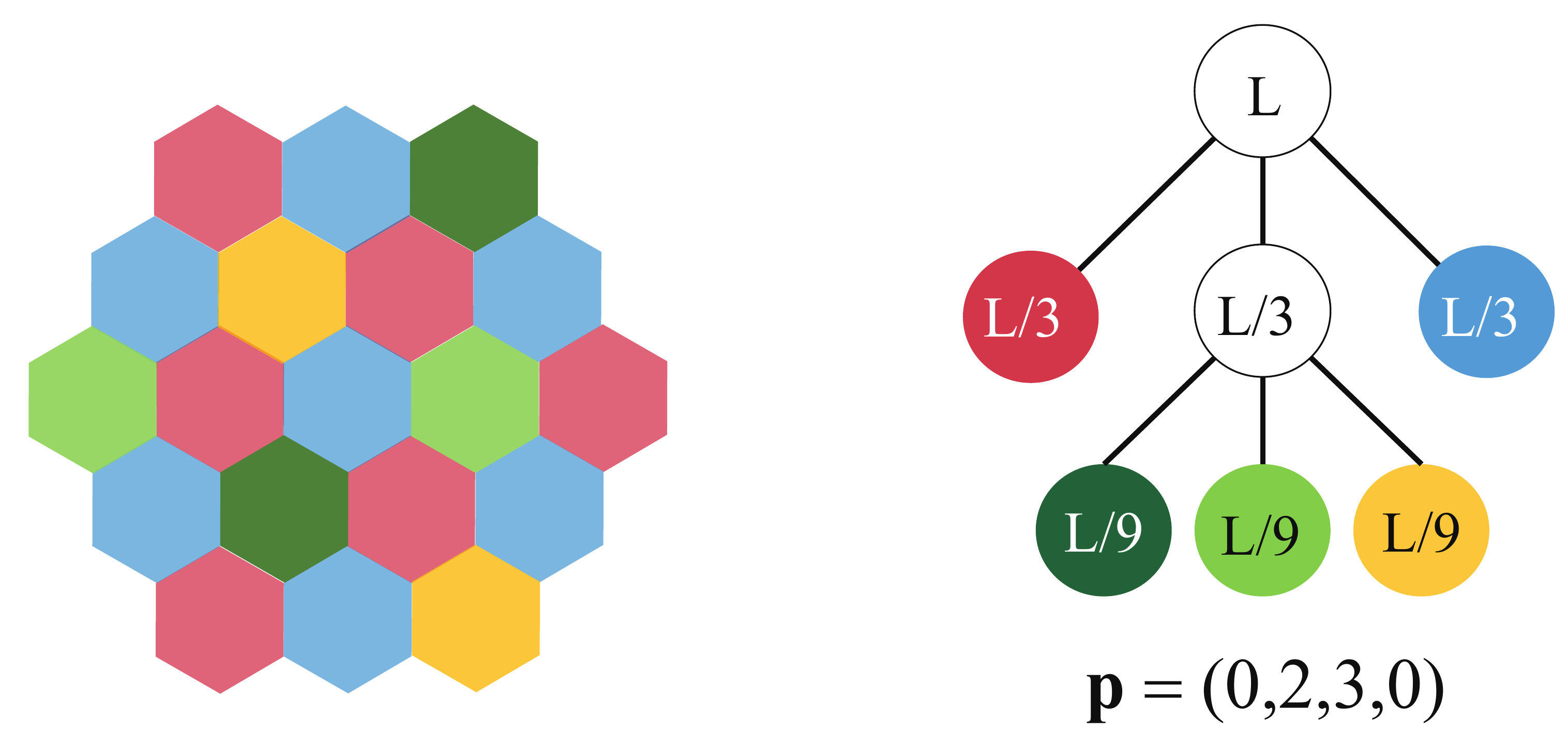}
	\caption{Example of coloring, tree structure, and pilot assignment vector}
	\label{Fig:coloring example}
	
\end{figure}

Notice that the users in different cells belonging to a given leaf experience pilot contamination. The severity of the contamination depends on 
the depth of the leaf. Every time the depth is increased, the distance between interfering cells (cells that reuse the pilot set)
increases by a certain factor, as can be observed from Fig. \ref{Fig:coloring example}. In fact, the distance grows geometrically as the depth increases.
According to (\ref{achievableR}), the achievable rate is determined by the $\beta$ values of the interfering users, which in turn depend on 
the distances of the interferers from the BS. 
The distance growth manifests itself as the reduced pilot contamination effect or an improved SIR, increasing the throughput. 
More specifically, the throughput grows roughly with $\log_2(\alpha^{2\gamma} SIR_1)=2\gamma \log_2(\alpha)+\log_2 SIR_1$, where 
$SIR_1$ is the reference SIR corresponding to the pilot reuse factor 1 and $\alpha$ is the parameter that represents the distance growth. It is clear that a geometric growth of 
$\alpha$ gives rise to a linear increase in the throughput. 
Letting $C_{i}$ be the rate of a user at depth $i$, this is to say that $C_{i}$ increases linearly with depth $i$.
With $\gamma = 3.7$ and $\alpha = \sqrt{3}$ (based on 3-way partitioning), we have $C_{i+1} \simeq C_{i} + 6$.

The per-cell sum rate for the network can be expressed as 
\begin{equation} \label{sum rate}
C_{sum}(\mathbf{p})=\frac{1}{L}\sum\limits_{i=0}^{\log_3 L-1} L3^{-i} p_i  C_{i}=\sum\limits_{i=0}^{\log_3 L-1} 3^{-i} p_i  C_{i}.
\end{equation}
The per-cell net sum rate, accounting for the fact that useful data gets transmitted only over the portion of the coherence time not allocated to the pilots,
is given by 
%\begin{eqnarray} \label{total transmitted bits}
\begin{align} \label{net sum rate}
C_{net}(\mathbf{p}) =\frac{N_{coh}-N_{pil}(\mathbf{p})}{N_{coh}} C_{sum}(\mathbf{p}).
\end{align}
%\end{eqnarray}
We shall use $C_{net}$ as the objective function for 
finding optimal pilot assignment strategies.

%% Closed-Form Solution for Optimal Pilot Assigning Strategy
\subsection{Closed-Form Solution of Optimal Pilot Assignment Strategy}\label{MathResults_InfiniteM}
The optimal pilot assignment vector $\mathbf{p}_{opt}$ for the given normalized coherence time $N_{coh}$, the number of cells $L$, and the number of users per cell $K$ is:
\begin{equation*}\label{optimal assigning}
\mathbf{p}_{opt} = \underset{\mathbf{p}\in P_{L,K}}{\arg\max}\ C_{net}(\mathbf{p}).
\end{equation*}
We also note that all valid pilot assignments yield pilot lengths of same parity with $K$, as formally stated in Lemma \ref{Lemma_pilot_length} (with proof given in Appendix A).

\begin{lemma} \label{Lemma_pilot_length}
For given $L$ and $K$, $\{N_{pil}(\mathbf{p}) : \mathbf{p} \in P_{L,K}\} = \{K,K+2,K+4,\cdots, \frac{LK}{3} \}$.
\end{lemma}

Before giving the first main theorem, it is useful to define the pilot assignment vector that maximizes the per-cell sum rate $C_{sum}$ with a finite pilot length constraint:
\begin{equation*}\label{optimal_prime}
\mathbf{p}'_{opt}(N_{p0})=\underset{\mathbf{p}\in \Omega(N_{p0})}{\arg\max} \ C_{sum}(\mathbf{p})
\end{equation*}
where 
%\begin{equation*}\label{Omega_variable}
$\Omega(N_{p0})=\{\mathbf{p}\in P_{L,K} \: | \: N_{pil}(\mathbf{p})=N_{p0}\}.$
%\end{equation*}

Also, it is useful to define a transition vector associated with each valid pilot assignment vector. Given a pilot assignment vector $\mathbf{p}$, a transition vector $\mathbf{t}$ is a vector whose $i$-th element $t_i$ represents the number of (3-way) partitioning acts taken place at depth $i$ as the full reuse assignment vector of
$(1,0,\cdots,0)$ transitions to $\mathbf{p}$. For example, $\mathbf{p}=(1,0,0,0)$ turns to $\mathbf{p}=(0,2,3,0)$ via a transition vector $\mathbf{t}=(1,1,0)$. The first transition element $t_0=1$ triggers a (3-way) partitioning at depth 0, temporarily creating a pilot vector (0,3,0,0). The next element $t_1=1$ induces a (3-way) partitioning on one of the 3 existing partitions at depth 1, thereby giving rise to a new pilot vector (0,2,3,0). Since the next transition vector element is zero, the partitioning stops. The transition elements also point to the number of white nodes at each depth, as can be confirmed in Fig. \ref{Fig:coloring example}. A general definition which relates the pilot assignment vector $\mathbf{p}$ and the corresponding transition vector $\mathbf{t}$ is given as follows.
%We will first find a closed-form solution for $\mathbf{p}'_{opt}$ and then find eventually $\mathbf{p}_{opt}$ by exploring its relationship with the former. 

\begin{definition}[Definition:\nopunct]
For each valid pilot assignment vector $\mathbf{p}=(p_{0},p_{1},\cdots,p_{\log_3 L-1}) \in P_{L,K}$, the corresponding transition vector $\mathbf{t}=(t_{0},t_{1},\cdots,t_{\log_3 L-2})$ is defined as:
\begin{equation}\label{ptot}
\begin{cases}
t_{0} = K - p_{0} \\
t_{i} = -p_{i} + 3t_{i-1}. & 1\leq i \leq  \log_3 L-2 \\
\end{cases}
\end{equation}
The inverse relationship exists:
\begin{equation}\label{ttop}
\begin{cases}
p_{0} = K - t_{0} \\
p_{i} = 3t_{i-1}-t_{i} & 1\leq i \leq  \log_3 L-2 \\
p_{\log_3 L-1}=3t_{\log_3 L-2}.
\end{cases}
\end{equation}
\end{definition}
The first two equations of (\ref{ttop}) come from (\ref{ptot}), while the last equation is from the fact that $\sum_{i=0}^{\log_3 L-1}p_{i}3^{-i}=K$ as in (\ref{pilot assgning vector}). A useful Lemma on the property of the transition vector is stated as follows (with proof given in  Appendix A).

\begin{lemma} \label{Lemma_transition_vector}
Any transition vector $\mathbf{t}=(t_{0},t_{1},\cdots,t_{\log_3 L-2})$ originated from $\mathbf{p} \in \Omega(N_{p0})$ satisfies
\begin{equation} \label{Lemma_transition_vector result}
\begin{cases}
0 \leq t_{i} \leq K3^{i} & 0 \leq i \leq \log_3 L-2 \\
\displaystyle\sum_{i=0}^{\log_3 L-2}t_{i}=\frac{N_{p0}-K}{2}. 
\end{cases}
\end{equation}
\end{lemma}

We further define the index function  $\chi(N_{p0})$ that identifies the first non-zero position of $\mathbf{p}'_{opt}(N_{p0})$. Note that for a given $N_{p0}$, $(N_{p0}-K)/2$ is the total number of partitioning acts to get from 
$(1,0,\cdots, 0)$ to an arbitrary $\mathbf{p} \in \Omega(N_{p0})$, as formally stated in the second equation of (\ref{Lemma_transition_vector result}). 
Since the maximum value of $t_i$ is $K3^{i}$,
and as the partitioning steps are applied from the top down for a given $N_{p0}$,
all nodes through depth $k-1$ will have been be partitioned 
if 
$\sum_{i=0}^{k-1}K3^{i} \leq (N_{p0}-K)/2$. Continuing to the next depth, 
however, only a portion of the nodes will have been partitioned at depth $k$,
in which case $\sum_{i=0}^{k}K3^{i} > (N_{p0}-K)/2$.
Recall the nodes that have not been partitioned are leaves,
and $\chi(N_{p0})$ is the shortest depth of any leaf.
With the leaf appearing first at depth $k$, 
we can formally write:
%\begin{equation*}\label{xi_index}
$\chi(N_{p0})=\min\{k \: | \: \sum_{i=0}^{k}K3^{i}>\frac{N_{p0}-K}{2} \}.$
%\end{equation*}

We will first lay out a closed-form solution for $\mathbf{p}'_{opt}$ in Theorem \ref{Theorem1} and then 
find eventually $\mathbf{p}_{opt}$ in Theorem \ref{Theorem2} by exploring its relationship with the former. The proof of Theorem 1 is given in the next subsection.

\begin{theorem} \label{Theorem1}
Given $L$, $K$, and $N_{p0}$, the optimal pilot assignment vector \\ $\mathbf{p}'_{opt}(N_{p0})=(p'_{0},\cdots,p'_{\log_3 L -1})$ with respect to $C_{sum}$, has its components as follows:
\begin{equation} \label{Thm1 result}
p'_{i} =
\begin{cases}
\displaystyle\sum_{s=0}^{i}K3^{s}-\frac{N_{p0}-K}{2} & i=\chi(N_{p0}) \\
3\left(\displaystyle\frac{N_{p0}-K}{2}-\displaystyle\sum_{s=0}^{i-2}K3^{s}  \right) & i=\chi(N_{p0})+1 \\
0 & \textrm{otherwise}
\end{cases}
\end{equation}
\end{theorem}

%Here we only provide a brief sketch of the proof due to space limitation. The sum rate in (\ref{sum rate}) can be expressed using the transition vector $\mathbf{t}$ via (\ref{ptot}):
%\begin{equation}
%C_{sum}(\mathbf{t})= C_{0} + \sum\limits_{i=0}^{\log_3 L-2}t_{i}3^{-i}(C_{i+1}-C_{i}).
%\end{equation}
%Because $C_{i}$ is a linear function of $i$, the difference $(C_{i+1}-C_{i})$ is a constant. Therefore, all we have to do is to find the optimal $\mathbf{t}$ which maximizes $\sum_{i=0}^{\log_3 L-2}t_{i}3^{-i}$.

%Under the condition of Lemma \ref{Lemma2}, for maximizing $\sum_{i=0}^{\log_3 L-2}t_{i}3^{-i}$, we should give the largest values $t_{i}=3^{i}$ to the lower indices $i<\chi(N_{p0})$ and the remaining values to the rest of the indices. The final solution for $\mathbf{t}'=(t_{0}',t_{1}',\cdots,t_{\log_3 L-2}')$ is:
%\begin{equation} \label{t_i'}
%t_{i}'=
%\begin{cases}
%3^{i} & i < \chi(N_{p0}) \\
%\frac{N_{p0}-1}{2} - \displaystyle\sum_{i=0}^{\chi(N_{p0}) - 1}3^{i}	& i= \chi(N_{p0}) \\
%0 & i > \chi(N_{p0}),
%\end{cases}
%\end{equation}
%which leads to (\ref{Thm1 result}) via (\ref{ttop}). 

For example, given $L=81$, $K=1$, and $N_{p0}=7$, $\chi(N_{p0})=1$ since $3^0  < (N_{p0}-K)/2 = 3 < 3^0 + 3^1 $. Therefore, 
$\mathbf{p}'_{opt}(7)$ has its components $p'_{0}=p'_{3}=0$, $p'_{1}=\sum_{s=0}^{1}3^{s}-(7-1)/2 = 1$, and $p'_{2}=3\{(7-1)/2 - \sum_{s=0}^{0}3^{s}\} = 6$, which result in $\mathbf{p}'_{opt}(7) = (0,1,6,0)$.

For a given $N_{p0}$, more than one valid pilot assignments may exist. For example, if $L=81$, $K=1$, and $N_{p0}=7$, $\mathbf{p}=(0,1,6,0)$ and $\mathbf{p}=(0,2,2,3)$ are valid vectors, but Theorem \ref{Theorem1} reveals that $\mathbf{p}'_{opt}(N_{p0})=(0,1,6,0)$. From Theorem \ref{Theorem1}, a certain trend relating the optimal pilot assignment vectors $\mathbf{p}'_{opt}(N_{p0})$ and  $\mathbf{p}'_{opt}(N_{p0}+2)$ can be observed as follows.

\begin{corollary} \label{Corollary1}
For two pilot lengths $N_{p0}$ and $N_{p0}+2$, the two corresponding optimal pilot assignment vectors $\mathbf{p}'_{opt}(N_{p0})=(p^{*}_{0},\cdots,p^{*}_{\log_3 L -1})$ and $\mathbf{p}'_{opt}(N_{p0}+2)=(p^{**}_{0},\cdots,p^{**}_{\log_3 L -1})$ exhibit the following relationship:
\begin{equation*}
p^{**}_{i} = 
\begin{cases}
p^{*}_{i}-1 & i=\chi(N_{p0}) \\
p^{*}_{i}+3 & i=\chi(N_{p0})+1 \\
p^{*}_{i} & \textrm{otherwise}
\end{cases}
\end{equation*}
\end{corollary}

\begin{proof}[Proof:\nopunct]
In the case  $\chi(N_{p0})=\chi(N_{p0}+2)$, Corollary \ref{Corollary1} is a direct consequence of  Theorem \ref{Theorem1}.
In other case, i.e., $\chi(N_{p0}+2)=\chi(N_{p0})+1$, we can see $\sum_{i=0}^{\chi(N_{p0})}K3^{i}=\frac{N_{p0}-K}{2}+1$ from the definition of the $\chi$ function. This coupled with Theorem \ref{Theorem1} proves Corollary \ref{Corollary1}.
\end{proof}

For a given $N_{p0}$, the optimal assignment vectors $\mathbf{p}'_{opt}(N_{p0})$ and $\mathbf{p}'_{opt}(N_{p0}+2)$ show a predictable pattern of tossing 1 from the left-most non-zero component to increase the adjacent component by 3. For example, in the case of $L=81$ and $K=1$, the optimal assignment $\mathbf{p}'_{opt}(7)=(0,1,6,0)$ can be transformed by reducing the second component by 1 and increasing the third one by 3, which results in the next optimal assignment for $N_{p0}=9$: $\mathbf{p}'_{opt}(9)=(0,0,9,0)$. It can be seen that there is a tendency to reduce the left-most non-zero values which give the most severe pilot contamination.

We now set out to find $\mathbf{p}_{opt}$.
Theorem \ref{Theorem1} and Corollary $\ref{Corollary1}$ already identify, given the fairly mild constraints of 
hexagonal cells and equi-distance partitioning, the optimal pilot assignment strategy maximizing the sum rate
for the chosen pilot sequence length. The next step is to find the relationship between the normalized coherence time $N_{coh}$ and the optimal pilot sequence length.

First, write the net sum-rate as 
%\begin{gather*} \label{eq24}
$C_{net}(\mathbf{p}'_{opt}(N_{p0})) =\frac{N_{coh}-N_{p0}}{N_{coh}}  C_{sum}(\mathbf{p}'_{opt}(N_{p0})),$
%\end{gather*}
which is an increasing function of $N_{coh}$ and crosses the horizontal axis once at $N_{coh}=N_{p0}$. Moreover this function saturates to $C_{sum}(\mathbf{p}'_{opt}(N_{p0}))$ for very large $N_{coh}$. Imagine plotting this function 
for all possible values of $N_{p0}=K,K+2,K+4,\cdots, LK/3$. As $N_{p0}$ increases, the zero-crossing is naturally shifted to the right while the saturation value moves up. More specifically, the $C_{net}$ curve for $N_{p0}=2(n-1)+K$ and 
that for $N_{p0}=2n+K$ intersect once. On the left side of this intersection point, the $C_{net}$ curve for $N_{p0}=2(n-1)+K$ is above the latter while on the right side, the latter curve is higher than the former. Let the horizontal value of this intersection point be $N_{coh}=\Delta_{n}$. 
%Applying the same argument, the two $C_{net}$ curves for $N_{p0}=2n+1$ and $N_{p0}=2n+3$ intersect at $\Delta_{n+1}$ somewhere to the right of the previous intersection point $\Delta_{n}$. In between $\Delta_{n}$ and $\Delta_{n+1}$, the $C_{net}$ curve for $N_{coh}=2n+1$ yields the highest values, indicating that when $N_{coh}$ falls between the two intersection points, the optimal pilot length is $2n+1$. 
It can be shown that (proof given in the next subsection) the intersection points are given by
\begin{equation}\label{delta_n}
\Delta_{n} = 
2\left( 2n-1-\displaystyle \sum_{i=0}^{\eta(n)-1} K3^{i}+K\xi(n) \right)+K
\end{equation} 
for $1\leq n \leq N_{L,K}$,
where $\eta(n) = \chi(2n+K-2)$ and $\xi(n)=3^{\eta(n)} C_{\eta(n)}/(C_{\eta(n)+1}-C_{\eta(n)})$, with $C_{i}$ already defined earlier in this section, and $N_{L,K}$ is the number of all possible pilot lengths minus one. Since $N_{pil}=K,K+2,\cdots,LK/3$ from Lemma \ref{Lemma_pilot_length}, we have $N_{L,K}=\frac{1}{2}(\frac{LK}{3} - K)$. We now state our second main theorem.

\begin{theorem} \label{Theorem2}
For given $L$, $K$, and $N_{coh}$, if $N_{coh}$ is in between two adjacent time points $\Delta_{n}$ and $\Delta_{n+1}$, i.e., $\Delta_{n}\leq N_{coh} <\Delta_{n+1}$, then the optimal assignment vector $\mathbf{p}_{opt}$ satisfies
\begin{gather*}
N_{pil}\left(\mathbf{p}_{opt}\right)=2n+K \nonumber\\
\mathbf{p}_{opt}=\mathbf{p}'_{opt}(2n+K). \label{Thm2 result}
\end{gather*}
Also, if $N_{coh} \geq \Delta_{N_{L,K}}$, then $\mathbf{p}_{opt}=(0,\cdots,0,LK/3)$.
\end{theorem}

% Section III-D
\subsection{Proofs of Theorems 1 and 2}

\subsubsection{Proof of Theorem 1}

The sum rate in (\ref{sum rate}) can be expressed using the transition vector $\mathbf{t}$ via (\ref{ttop}):
\begin{equation*}
C_{sum}(\mathbf{t})= KC_{0} + \sum\limits_{i=0}^{\log_3 L-2}t_{i}3^{-i}(C_{i+1}-C_{i}).
\end{equation*}
Because $C_{i}$ is a linear function of $i$, the difference $(C_{i+1}-C_{i})$ is a constant. Therefore, all we have to do is to find the optimal $\mathbf{t}$ which maximizes $\sum_{i=0}^{\log_3 L-2}t_{i}3^{-i}$.

Write the transition vector for the pilot assignment vector of (\ref{Thm1 result}):
\begin{equation} \label{pf_t_i'}
t_{i}'=
\begin{cases}
K3^{i} & i < \chi(N_{p0}) \\
\frac{N_{p0}-K}{2} - \displaystyle\sum_{i=0}^{\chi(N_{p0}) - 1}K3^{i}	& i= \chi(N_{p0}) \\
0 & i > \chi(N_{p0}).
\end{cases}
\end{equation}
Then, it can be shown that $\sum_{i=0}^{log_3 L-2}t_{i}'\cdotp3^{-i}$ is always greater than $\sum_{i=0}^{log_3 L-2}t_{i}\cdotp3^{-i}$ for any other transition vectors $\mathbf{t}=(t_{0},t_{1},\cdots,t_{log_3 L-2})$ corresponding to  $\mathbf{p} \in \Omega(N_{p0})$:

\textbf{Case 1:} if $t_{\chi(N_{p0})}' \geq t_{\chi(N_{p0})}$,
\begin{align*}
\delta  & \triangleq \sum\limits_{i=0}^{log_3 L-2} t_{i}'\cdotp3^{-i} -\sum\limits_{i=0}^{log_3 L-2} t_{i} \cdotp3^{-i} %= \sum\limits_{i=0}^{log_3 L-2}(t_{i}'-t_{i})\cdotp3^{-i} 
\nonumber\\
&= \sum\limits_{i=0}^{\chi(N_{p0})}(t_{i}'-t_{i})\cdotp3^{-i}+\sum\limits_{i=\chi(N_{p0})+1}^{log_3 L-2}(0-t_{i})\cdotp3^{-i} \nonumber\\
& \geq \left(\sum\limits_{i=0}^{\chi(N_{p0})}(t_{i}'-t_{i})\right)\cdotp3^{-\chi(N_{p0})} %+\sum\limits_{i=\chi(N_{p0})+1}^{log_3 L-2}(0-t_{i})\cdotp3^{-i} \nonumber\\
%& \geq \left(\sum\limits_{i=0}^{\chi(N_{p0})}(t_{i}'-t_{i})\right)\cdotp3^{-\chi(N_{p0})}
\nonumber\\
& \quad \quad \quad +\left(\sum\limits_{i=\chi(N_{p0})+1}^{log_3 L-2}(0-t_{i})\right)\cdotp3^{-(\chi(N_{p0})+1)}\nonumber\\
& = \left(\sum\limits_{i=0}^{\chi(N_{p0})}(t_{i}'-t_{i})\right)\cdotp(3^{-\chi(N_{p0})}-3^{-(\chi(N_{p0})+1)}) \geq 0 
\end{align*}
where the first inequality holds due to $0\leq t_{i}\leq K3^{i}=t_{i}'$ for $0\leq i < \chi(N_{p0})$ (by Lemma \ref{Lemma_transition_vector}), the assumption of $t_{\chi(N_{p0})}' \geq t_{\chi(N_{p0})}$ , and the fact that $t_{i} \geq 0$ for $0\leq i \leq log_3 L-2 $ (by Lemma \ref{Lemma_transition_vector}).
The last equality holds because $\sum_{i=0}^{\log_3 L-2}t_{i}=\frac{N_{p0}-K}{2}$ for any valid transition vector $\mathbf{t}=(t_{0},t_{1},\cdots,t_{log_3 L-2})$ by Lemma \ref{Lemma_transition_vector}. %(note that $\mathbf{t}'=(t_{0}',t_{1}',\cdots,t_{log_3 L-2}')$ is also a valid transition vector) 
Similarly,

\textbf{Case 2:} if $t_{\chi(N_{p0})}' < t_{\chi(N_{p0})}$,
\begin{align*}
\delta  & \triangleq \sum\limits_{i=0}^{log_3 L-2} t_{i}'\cdotp3^{-i} -\sum\limits_{i=0}^{log_3 L-2} t_{i} \cdotp3^{-i}  %= \sum\limits_{i=0}^{log_3 L-2}(t_{i}'-t_{i})\cdotp3^{-i} \nonumber\\
\nonumber\\
&= \sum\limits_{i=0}^{\chi(N_{p0})-1}(t_{i}'-t_{i})\cdotp3^{-i}+\sum\limits_{i=\chi(N_{p0})}^{log_3 L-2}(t_{i}'-t_{i})\cdotp3^{-i} \nonumber\\
& \geq \left(\sum\limits_{i=0}^{\chi(N_{p0})-1}(t_{i}'-t_{i})\right)\cdotp3^{-(\chi(N_{p0})-1)} %+\sum\limits_{i=\chi(N_{p0})}^{log_3 L-2}(t_{i}'-t_{i})\cdotp3^{-i} \nonumber\\
%& \geq \left(\sum\limits_{i=0}^{\chi(N_{p0})-1}(t_{i}'-t_{i})\right)\cdotp3^{-(\chi(N_{p0})-1)}
\nonumber\\
& \quad \quad \quad +\left(\sum\limits_{i=\chi(N_{p0})}^{log_3 L-2}(t_{i}'-t_{i})\right)\cdotp3^{-\chi(N_{p0})}\nonumber\\
& = \left(\sum\limits_{i=0}^{\chi(N_{p0})-1}(t_{i}'-t_{i})\right)\cdotp(3^{-(\chi(N_{p0})-1)}-3^{-\chi(N_{p0})}) \geq 0.
\end{align*}
%Similarly, the first inequality holds due to $0\leq t_{i}\leq 3^{i}=t_{i}'$ for $0\leq i < \chi(N_{p0})$ (by Lemma \ref{Lemma_transition_vector}), $t_{i} \geq 0$ for $0\leq i \leq log_3 L-2 $ (by Lemma \ref{Lemma_transition_vector}), and the assumption of $t_{\chi(N_{p0})}' < t_{\chi(N_{p0})}$. 
%The last equality holds because $\sum_{i=0}^{\log_3 L-2}t_{i}=\frac{N_{p0}-1}{2}$ for any valid transition vector $\mathbf{t}=(t_{0},t_{1},\cdots,t_{log_3 L-2})$ by Lemma \ref{Lemma_transition_vector} (note that $\mathbf{t}'=(t_{0}',t_{1}',\cdots,t_{log_3 L-2}')$ is also a valid transition vector). 
  
Thus, we see that $\delta = 0$ iff $\mathbf{t}=\mathbf{t}'$, due to the inherent property of a transition vector (second equation in Lemma \ref{Lemma_transition_vector}).
Therefore, (\ref{pf_t_i'}) gives the maximum $C_{sum}$ value among all valid transition vectors.
Applying the inverse transformation (\ref{ttop}) to (\ref{pf_t_i'}) proves Theorem \ref{Theorem1}.

\subsubsection{Proof of Theorem 2}

\begin{figure}[!t]
	\centering
	\includegraphics[height=30mm]{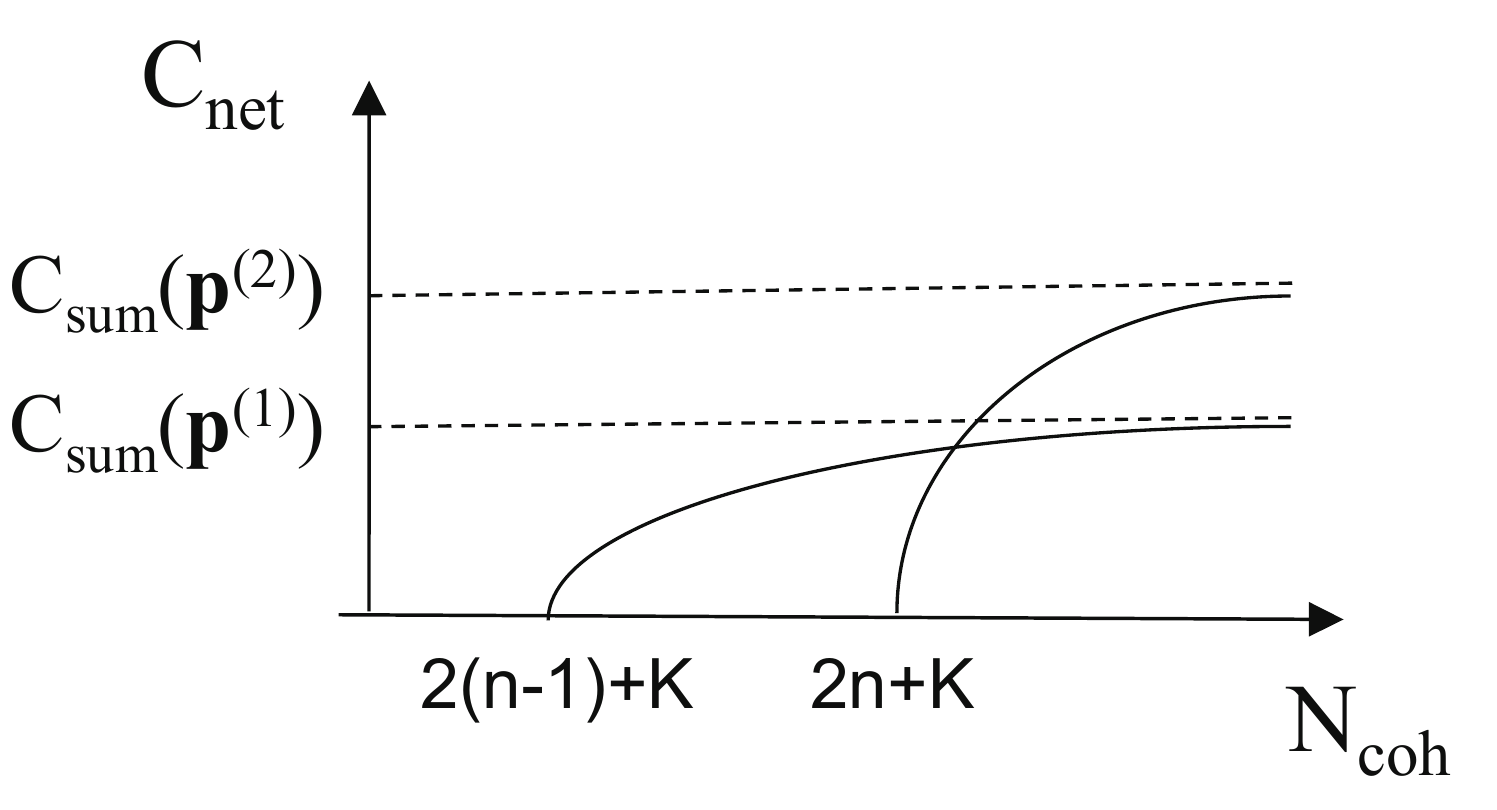}
	\caption{Graph of $C_{net}$ for two adjacent $N_{p0}$ values}
	\label{Fig:Theorem2_proof_2}
	
\end{figure}

%%% proof of C_sum increase as \tau increases (using corollary 1)
To minimize the notational cluttering, denote $\mathbf{p}^{(1)} = \mathbf{p}'_{opt}(2(n-1)+K)$, $\mathbf{p}^{(2)} = \mathbf{p}'_{opt}(2n+K)$ and $\mathbf{p}^{(3)} = \mathbf{p}'_{opt}(2(n+1)+K)$. First, we wish to find the relationship between $C_{sum}(\mathbf{p}^{(1)})$ and $C_{sum}(\mathbf{p}^{(2)})$ to get the expression (\ref{delta_n}) for $\Delta_{n}$. 
\begin{align*}
%  \phantom{\mathllap{C_{sum}(\mathbf{p}^{*},L) }}
%  \begin{aligned}
&C_{sum}(\mathbf{p}^{(2)})
= \sum\limits_{i=0}^{\log_3 L-1} 3^{-i} p_{i}^{(2)}  C_{i}\nonumber\\
&= \left( \sum _{\substack {i \neq \eta(n) \\ i \neq \eta(n) + 1 }} 3^{-i} p_{i}^{(1)}  C_{i} \right)  + 3^{-\eta(n)} (p_{\eta(n)}^{(1)} - 1) C_{\eta(n)} \nonumber\\
&\quad \quad \quad + 3^{-(\eta(n) + 1)} (p_{\eta(n)+1}^{(1)} +3) C_{\eta(n)+1}\nonumber\\  
%&= \sum\limits_{i=0}^{\log_3 L-1} 3^{-i} p_{i}^{(1)}  C_{i} + 3^{-\eta(n)}(C_{\eta(n)+1}- C_{\eta(n)}  ) 
&= C_{sum}(\mathbf{p}^{(1)}) + 3^{-\eta(n)}(C_{\eta(n)+1}- C_{\eta(n)}  )\nonumber\\
&\simeq C_{sum}(\mathbf{p}^{(1)}) + 6\cdot3^{-\eta(n)} > C_{sum}(\mathbf{p}^{(1)})
%  \end{aligned}\nonumber\\
\end{align*}
where the first equality is from the definition of $C_{sum}$, the second equality is from Corollary \ref{Corollary1}. The last approximation came from the mathematical analysis for $C_{i}$ in Section III-B. It is now clear  that the curve for $\mathbf{p}^{(2)}$ has a larger saturation value than the curve for $\mathbf{p}^{(1)}$, as illustrated in Fig. \ref{Fig:Theorem2_proof_2}.

% Suggest intersection point of two curves in Fig.9 - connect to delta_n

Now, to get to the expression for the intersection point $\Delta_{n}$ in Fig. \ref{Fig:Theorem2_proof_2}, we start with $C_{net}(\mathbf{p}^{(1)}) = C_{net}(\mathbf{p}^{(2)})$ for $N_{coh}=\Delta_{n}$, developing a series of equalities:
\begin{align}
%  \phantom{i + j + k}
%  &\begin{aligned}
&\frac{\Delta_{n} - (2(n-1)+K)}{\Delta_{n}}C_{sum}(\mathbf{p}^{(1)}) \nonumber\\
= & \frac{\Delta_{n} - (2n+K)}{\Delta_{n}}C_{sum}(\mathbf{p}^{(2)}), \nonumber\\
&(\Delta_{n} - 2n - K + 2)C_{sum}(\mathbf{p}^{(1)}) \nonumber\\
= &(\Delta_{n} - 2n - K) \{C_{sum}(\mathbf{p}^{(1)}) + 3^{-\eta(n)} (C_{\eta(n)+1} - C_{\eta(n)})\},\nonumber\\ 
&2C_{sum}(\mathbf{p}^{(1)}) \nonumber\\
= &(\Delta_{n} - 2n - K) (C_{\eta(n)+1} - C_{\eta(n)}) 3^{-\eta(n)} 
%  \end{aligned}
\end{align}
leading finally to
\begin{align}\label{proof2_Delta_n}
\Delta_{n} %&= 2n+K +\frac{2C_{sum}(\mathbf{p}^{(1)})}{3^{-\eta(n)}(C_{\chi(2n-1)+1} - C_{\chi(2n-1)})} 
&= 2n+K +\frac{2\sum\limits_{i=0}^{\log_3 L-1} 3^{-i} p_{i}^{(1)}  C_{i}}{3^{-\eta(n)}(C_{\eta(n)+1} - C_{\eta(n)})}.
\end{align}
However, by inserting $N_{p0}=2n+K-2$ into (\ref{Thm1 result}), we can express $\sum\limits_{i=0}^{\log_3 L-1} 3^{-i} p_{i}^{(1)}  C_{i}$ as:
\begin{align}\label{proof2_part}
&\sum\limits_{i=0}^{\log_3 L-1} 3^{-i} p_{i}^{(1)}  C_{i} \nonumber\\
&= 3^{-\eta(n)} \left(n-1 - \sum\limits_{s=0}^{\eta(n)-1} K3^{s}\right) \left(C_{\eta(n)+1}-C_{\eta(n)}\right) \nonumber\\
& \quad \quad + KC_{\eta(n)}.
\end{align}
Inserting (\ref{proof2_part}) into (\ref{proof2_Delta_n}) results in (\ref{delta_n}).

% Prove that every curve has "range of N_coh" where they are maximum. 
As the net rate curves for $\mathbf{p}^{(1)}$ and $\mathbf{p}^{(2)}$ cross at $N_{coh}=\Delta_{n}$, the curves for $\mathbf{p}^{(2)}$ and $\mathbf{p}^{(3)}$ intersect at $N_{coh}=\Delta_{n+1}$. We shall now prove that $\Delta_{n}$ is a monotonically increasing sequence. First, note $\Delta_{n+1} - \Delta_{n} = 4 + 2 \gamma(n)$
where $\gamma(n) = \xi(n+1)-\xi(n) - ( \sum_{i=0}^{\eta(n+1)-1}3^{i} - \sum_{i=0}^{\eta(n)-1}3^{i} )$.  
In case of $\eta(n) = \eta(n+1)$, we have $\xi(n+1) = \xi(n)$, so that $\Delta_{n+1} - \Delta_{n}=4>0$.
%\begin{align}
%\xi(n+1) &= \frac{3^{\eta(n+1)} C_{\eta(n+1)}}{C_{\eta(n+1)+1}-C_{\eta(n+1)}} =\frac{3^{\eta(n)} C_{\eta(n)}}{C_{\eta(n)+1}-C_{\eta(n)}}=\xi(n).
%\end{align}
%Therefore, $\eta(n)=0$, and $\Delta_{n+1} - \Delta_{n}=4>0$
For other cases, i.e., $\eta(n+1) = \eta(n) + 1$, we have $\xi(n+1) = \frac{3^{\eta(n)+1} C_{\eta(n)+1}}{C_{\eta(n)+2}-C_{\eta(n)+1}}$, so that 
%\begin{align}
%\xi(n+1) &= \frac{3^{\eta(n+1)} C_{\eta(n+1)}}{C_{\eta(n+1)+1}-C_{\eta(n+1)}}=\frac{3^{\eta(n+1)} C_{\eta(n)+1}}{C_{\eta(n)+2}-C_{\eta(n)+1}},\nonumber\\
%\end{align}
\begin{align*}
\gamma(n) &= \frac{3^{\eta(n)+1} C_{\eta(n)+1}}{C_{\eta(n)+2}-C_{\eta(n)+1}}  - \frac{3^{\eta(n)} C_{\eta(n)}}{C_{\eta(n)+1}-C_{\eta(n)}} - 3^{\eta(n)-1}\nonumber\\
&\simeq \frac{3^{\eta(n)+1} (C_{\eta(n)}+6)}{6} -\frac{3^{\eta(n)} C_{\eta(n)}}{6}- 3^{\eta(n)} \nonumber\\
&= 2 \cdot 3^{\chi(2n-1)} \left(1+\frac{C_{\chi(2n-1)}}{6}\right) > 0.
\end{align*} 
The last approximation came from the mathematical analysis for $C_{i}$ in Section III-B. 
Therefore, in both cases, $\Delta_{n+1} - \Delta_{n} > 0.$
\begin{figure}[!t]
	\centering
	\includegraphics[width=85mm]{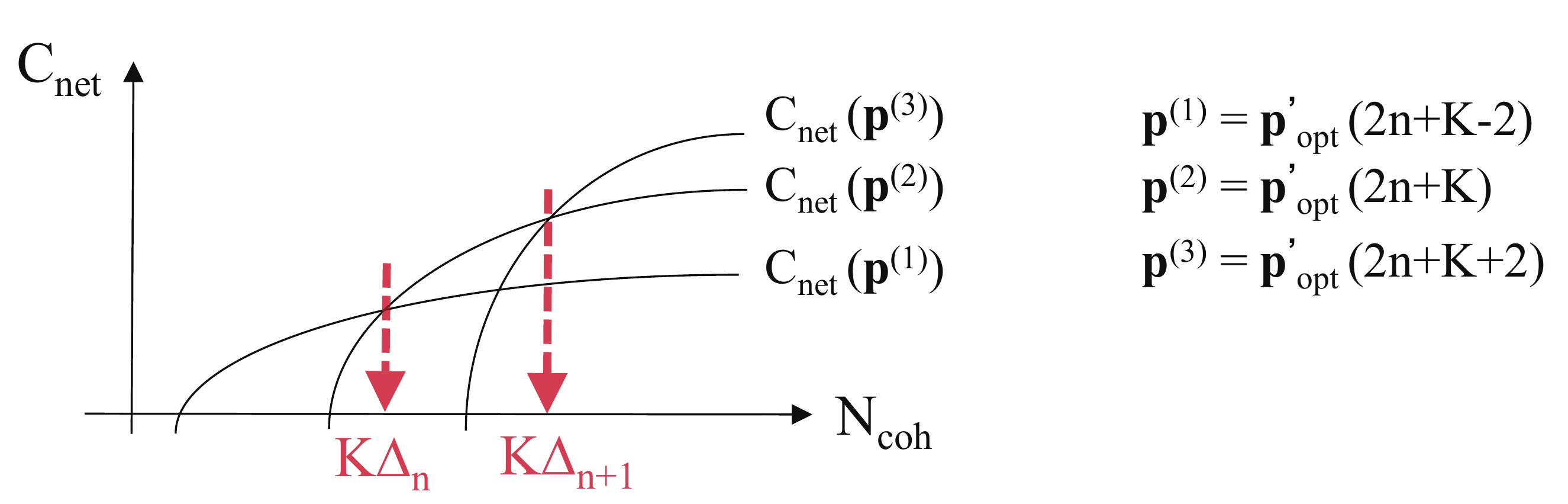}
	\caption{Graph of $C_{net}$ for three adjacent $N_{p0}$ values}
	\label{Fig:Theorem2_proof_3}
\end{figure}

Since $\Delta_n$ is an increasing sequence, we can complete a picture as 
illustrated in Fig. \ref{Fig:Theorem2_proof_3}. 
For $N_{coh}>\Delta_{n+1}$, $C_{net}(\mathbf{p}^{(3)})$ is greater than 
any net rate curves that will intersect with itself on the right side of 
$\Delta_{n+1}$. For $N_{coh}<\Delta_{n}$, $C_{net}(\mathbf{p}^{(1)})$ is greater than 
any net rate curves that will intersect itself on the left side of 
$\Delta_{n}$. Therefore, it is clear that in the interval $\Delta_{n}\leq N_{coh} < \Delta_{n+1}$, $C_{net}(\mathbf{p}^{(2)})$ is the largest of all net rate curves
and 
$\mathbf{p}^{(2)}=\mathbf{p}^{'}_{opt}(2n+K)$ is optimal among all possible pilot assignment vectors. 

For the boundary case, if $N_{coh} \geq \Delta_{N_{L,K}}$, then $\mathbf{p}_{opt}=\mathbf{p}^{'}_{opt}(LK/3)$. However, $\mathbf{p}^{'}_{opt}(LK/3)=(0,\cdots,0,LK/3)$ since $\Omega(LK/3)=\{(0,\cdots,0,LK/3)\}$ (note the analysis in Step 2 of the proof for Lemma 1). Therefore, $\mathbf{p}_{opt}=(0\cdots,0,LK/3)$ for $N_{coh} \geq \Delta_{N_{L,K}}$.
This completes the proof.

% Simulation Result
\section{Numerical Results}\label{Section:NumResult}
Simulation is necessary to obtain the $C_i$ values in (\ref{sum rate}). To get $C_i$, the $\beta$ terms in (\ref{achievableR}) need to be generated pseudo-randomly according to the assumed underlying statistical properties. The simulation is based on the following system parameters, which are consistent with those in \cite{bjornson2016massive}.  
We assume a signal decay exponent of $\gamma = 3.7$, a cell radius of $r$ meters and a cell-hole radius of $0.14r$. Note that $\{C_i\}$ does not depend on the $r$ value.
The number of cells are fixed to $L=81$, and the user locations are uniform-random within a cell. 
To generate each $C_i$ value, an average is taken over 100,000 pseudo-random trials.
Once the $C_i$ values are obtained, the optimal pilot length and the pilot assignment vector as well as the net throughput can be found for various coherence intervals $N_{coh}$. In this section, simulation results on the optimal pilot assignment rule and maximum net throughput is presented. We begin with single-user ($K=1$) case, and then analyze general multi-user ($K>1$) case.

\subsection{Single-user case}

%\begin{table}
%	\small
%\caption{Optimal pilot assignment ($L=81,K=1$)}
%\centering
%\label{Table:Optimal assignment}
%\begin{tabular}{|c|c|c|}
%\hline
%$N_{coh}$ & $\mathbf{p}_{opt}$  & $N_{pil}(\mathbf{p}_{opt})$\tabularnewline
%\hline
%$0\sim 4$ & $(1,0,0,0)$ & $1$\tabularnewline
%$5\sim 17$ & $(0,3,0,0)$ & $3$\tabularnewline
%$18\sim 21$ & $(0,2,3,0)$ & $5$\tabularnewline
%$22\sim 25$ & $(0,1,6,0)$ & $7$\tabularnewline
%$26\sim 68$ & $(0,0,9,0)$ & $9$\tabularnewline
%$69\sim 72$ & $(0,0,8,3)$ & $11$\tabularnewline
%$\vdots$ & $\vdots$ & $\vdots$\tabularnewline
%%$75\sim 78$ & $(0,0,7,6)$ & $13$\tabularnewline
%%$79\sim 82$ & $(0,0,6,9)$ & $15$\tabularnewline
%%$83\sim 86$ & $(0,0,5,12)$ & $17$\tabularnewline
%%$87\sim 90$ & $(0,0,4,15)$ & $19$\tabularnewline
%%$91\sim 94$ & $(0,0,3,18)$ & $21$\tabularnewline
%%$95\sim 98$ & $(0,0,2,21)$ & $23$\tabularnewline
%%$99\sim 102$ & $(0,0,1,24)$ & $25$\tabularnewline
%$101\sim$ & $(0,0,0,27)$ & $27$\tabularnewline
%\hline
%\end{tabular}
%\end{table}

\begin{table}
	\small
	\centering
	\caption{Optimal pilot assignment for $L=81$}
	\subfloat[$K=1$]{%
		\label{Table:Optimal assignment}
		%		\hspace{.5cm}%
		\begin{tabular}{|c|c|c|}
			\hline
			$N_{coh}$ & $\mathbf{p}_{opt}$  & $N_{pil}(\mathbf{p}_{opt})$\tabularnewline
			\hline
			$0\sim 4$ & $(1,0,0,0)$ & $1$\tabularnewline
			$5\sim 17$ & $(0,3,0,0)$ & $3$\tabularnewline
			$18\sim 21$ & $(0,2,3,0)$ & $5$\tabularnewline
			$22\sim 25$ & $(0,1,6,0)$ & $7$\tabularnewline
			$26\sim 68$ & $(0,0,9,0)$ & $9$\tabularnewline
			$69\sim 72$ & $(0,0,8,3)$ & $11$\tabularnewline
			$\vdots$ & $\vdots$ & $\vdots$\tabularnewline
			%$75\sim 78$ & $(0,0,7,6)$ & $13$\tabularnewline
			%$79\sim 82$ & $(0,0,6,9)$ & $15$\tabularnewline
			%$83\sim 86$ & $(0,0,5,12)$ & $17$\tabularnewline
			%$87\sim 90$ & $(0,0,4,15)$ & $19$\tabularnewline
			%$91\sim 94$ & $(0,0,3,18)$ & $21$\tabularnewline
			%$95\sim 98$ & $(0,0,2,21)$ & $23$\tabularnewline
			%$99\sim 102$ & $(0,0,1,24)$ & $25$\tabularnewline
			$101\sim$ & $(0,0,0,27)$ & $27$\tabularnewline
			\hline
			\end{tabular}%
		}\hspace{1cm}
		\subfloat[$K=2$]{%
			\label{Table:Optimal assignment2}
			\begin{tabular}{|c|c|c|}
				\hline
				$N_{coh}$ & $\mathbf{p}_{opt}$  & $N_{pil}(\mathbf{p}_{opt})$\tabularnewline
				\hline
				$0\sim 6$ & $(2,0,0,0)$ & $2$\tabularnewline
				$7\sim 10$ & $(1,3,0,0)$ & $4$\tabularnewline
				$11\sim 32$ & $(0,6,0,0)$ & $6$\tabularnewline
				$33\sim 36$ & $(0,5,3,0)$ & $8$\tabularnewline
				$37\sim 40$ & $(0,4,6,0)$ & $10$\tabularnewline
				$\vdots$ & $\vdots$ & $\vdots$\tabularnewline
				$203\sim $ & $(0,0,0,54)$ & $54$\tabularnewline
				\hline
		\end{tabular}%
	}
\end{table}

\begin{figure}
	\centering
	\subfloat[][$K=1$]{\includegraphics[width=60mm]{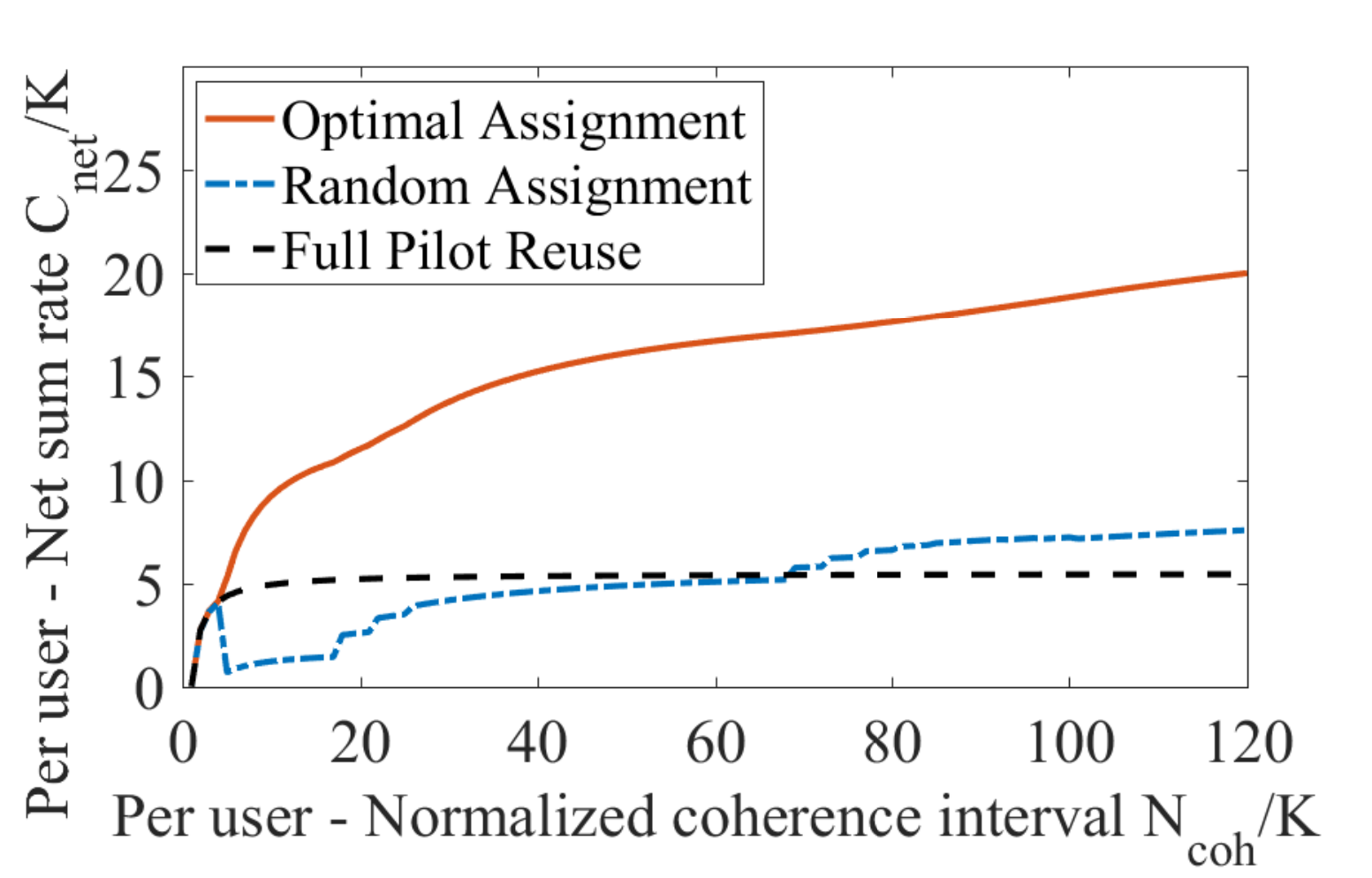}\label{Fig:K=1 Simulation}}
	\quad \quad
	\subfloat[][$K=14$]{\includegraphics[width=55mm]{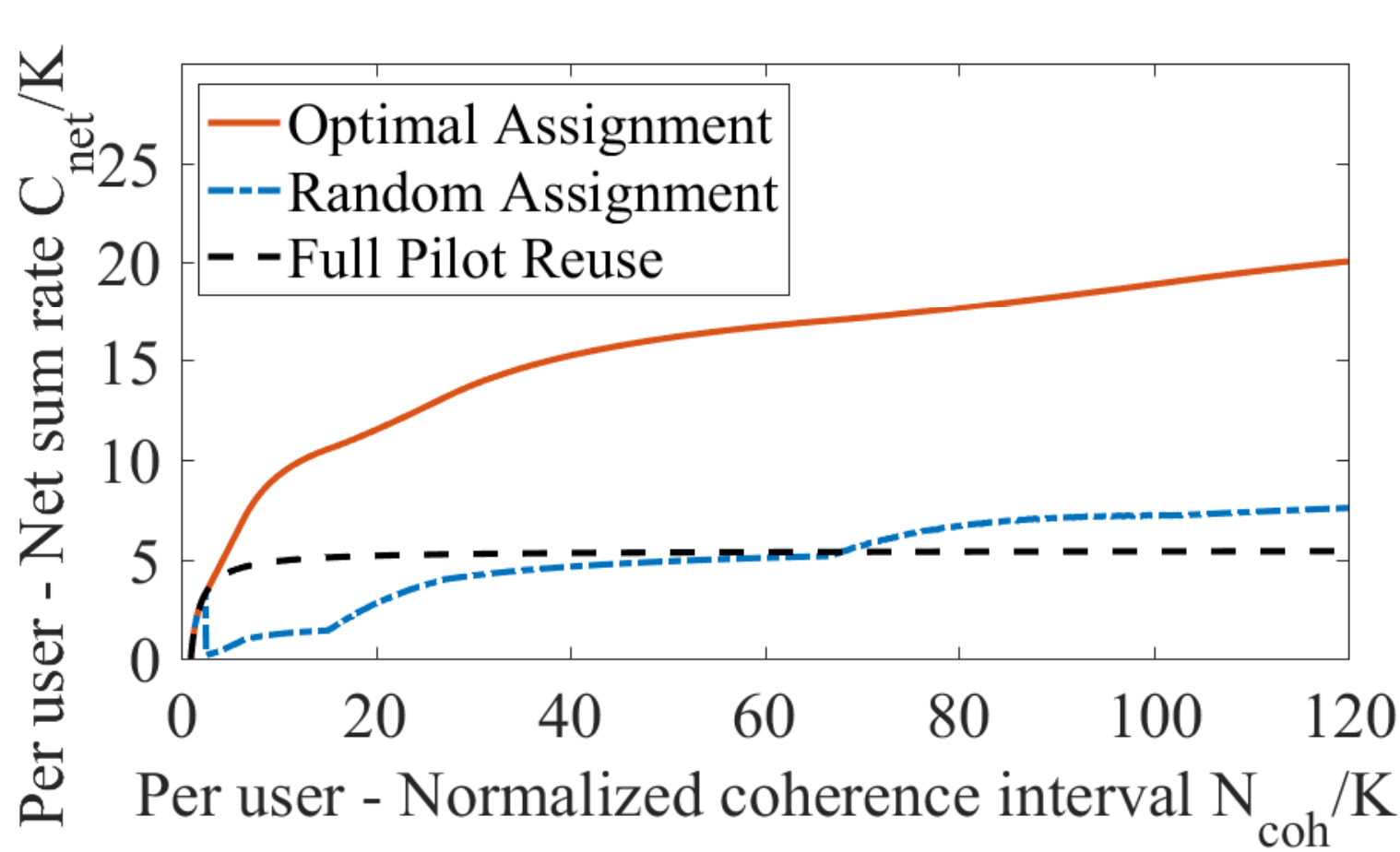}\label{Fig:K=14 Simulation}}
	\caption{Net rates of various pilot assignments for $L=81$}
\end{figure}

%\begin{figure}[!t]
%\centering
%        \includegraphics[height=50mm]{Fig6TotalRand1e51e2.eps}
%    \caption{Net rate for various pilot assignments ($L=81,K=1$)}
%    \label{Fig:K=1 Simulation} 
%\end{figure}

 Table \ref{Table:Optimal assignment} shows the optimal pilot assignment results for various values of $N_{coh}$.
We confirm that the simulation results of Table \ref{Table:Optimal assignment} are consistent with 
the mathematical analysis given in Theorem \ref{Theorem2}.
As $N_{coh}$ increases, the optimal pilot sequence gradually becomes longer. Moreover, the optimal assignment vectors show a pattern of tossing 1 from the left most non-zero component to increase the adjacent component by 3, for an example $(0,3,0,0) \rightarrow (0,2,3,0)$ (consistent with Corollary $\ref{Corollary1}$).

Fig. \ref{Fig:K=1 Simulation} shows the average achievable net rates for various pilot assignments versus normalized coherence interval. The random assignment means that a pilot sequence is chosen 
randomly and independently from $N_{pil}(\mathbf{p}_{opt})$ orthogonal pilots, and assigned to each user.
Therefore, the optimal assignment and the random assignment use the same amount of pilots for any given $N_{coh}$.
It can be seen that a substantial performance gain is obtained using the optimal method compared to the full pilot reuse case as $N_{coh}$ increases 
beyond 5. The random assignment is worse than the full reuse initially but eventually outperforms the latter as $N_{coh}$ grows.

For $N_{coh}=10,20$ and $40$, for example, the optimal assignment method has $87\%,121\%$ and $185\%$ higher net rates $C_{net}$ than the full pilot reuse assignment, respectively.
As coherence time increases, the benefit of allocating more time for pilots is considerable. 
Also, the non-shrinking performance gap between the optimal assignment and the random assignment indicate that
structured optimal assignment is required for a given pilot time, in order to maximize the net throughput of the network.

%% Simulation Result (MU)
\subsection{Multi-user case}
\label{Multiuser_numerical result}

%\begin{table}
%	\small
%\caption{Optimal pilot assignment ($L=81,K=2$)}
%\centering
%\label{Table:Optimal assignment2}
%\begin{tabular}{|c|c|c|}
%\hline
%$N_{coh}$ & $\mathbf{p}_{opt}$  & $N_{pil}(\mathbf{p}_{opt})$\tabularnewline
%\hline
%$0\sim 6$ & $(2,0,0,0)$ & $2$\tabularnewline
%$7\sim 10$ & $(1,3,0,0)$ & $4$\tabularnewline
%$11\sim 32$ & $(0,6,0,0)$ & $6$\tabularnewline
%$33\sim 36$ & $(0,5,3,0)$ & $8$\tabularnewline
%$37\sim 40$ & $(0,4,6,0)$ & $10$\tabularnewline
%$\vdots$ & $\vdots$ & $\vdots$\tabularnewline
%$203\sim $ & $(0,0,0,54)$ & $54$\tabularnewline
%\hline
%\end{tabular}
%\end{table}

In Table \ref{Table:Optimal assignment2}, the optimal assignment vectors and pilot lengths for various $N_{coh}$ are shown, assuming $L=81$ and $K=2$.
Like in the case for $K=1$, the optimal pilot assignment vectors have a predictable form (consistent with Corollary $\ref{Corollary1}$ and Theorem $\ref{Theorem2}$).

%\begin{figure}[!t]
%\centering
%        \includegraphics[height=50mm]{Fig7_Total(rand_1e5_1e2).png}
%    \caption{Net rate for various pilot assignments ($L=81,K=14$)}
%    \label{Fig:K=14 Simulation}
%\end{figure}

Fig. \ref{Fig:K=14 Simulation} shows the plots of the per-user net rates $C_{net}/K$
versus 
$N_{coh}/K$. These results are actually found for $K=14$, but very similar numerical results 
are obtained for the optimal scheme irrespective of the particular values of $K$ while
the results for the full pilot reuse are identical across all values of $K$. As a case in point,
it can be seen that the plots in Fig. \ref{Fig:K=1 Simulation} obtained for 
$K=1$ are nearly identical to those in Fig. \ref{Fig:K=14 Simulation} corresponding to $K=14$.
%(mathematical justifications will be given for this similarity in the following subsection).
Using the optimal pilot assignment scheme,
the per-user net rate improves 
substantially with increasing $N_{coh}/K$, relative to the full reuse scheme.

\begin{figure}[!t]
	\centering
	\includegraphics[height=50mm]{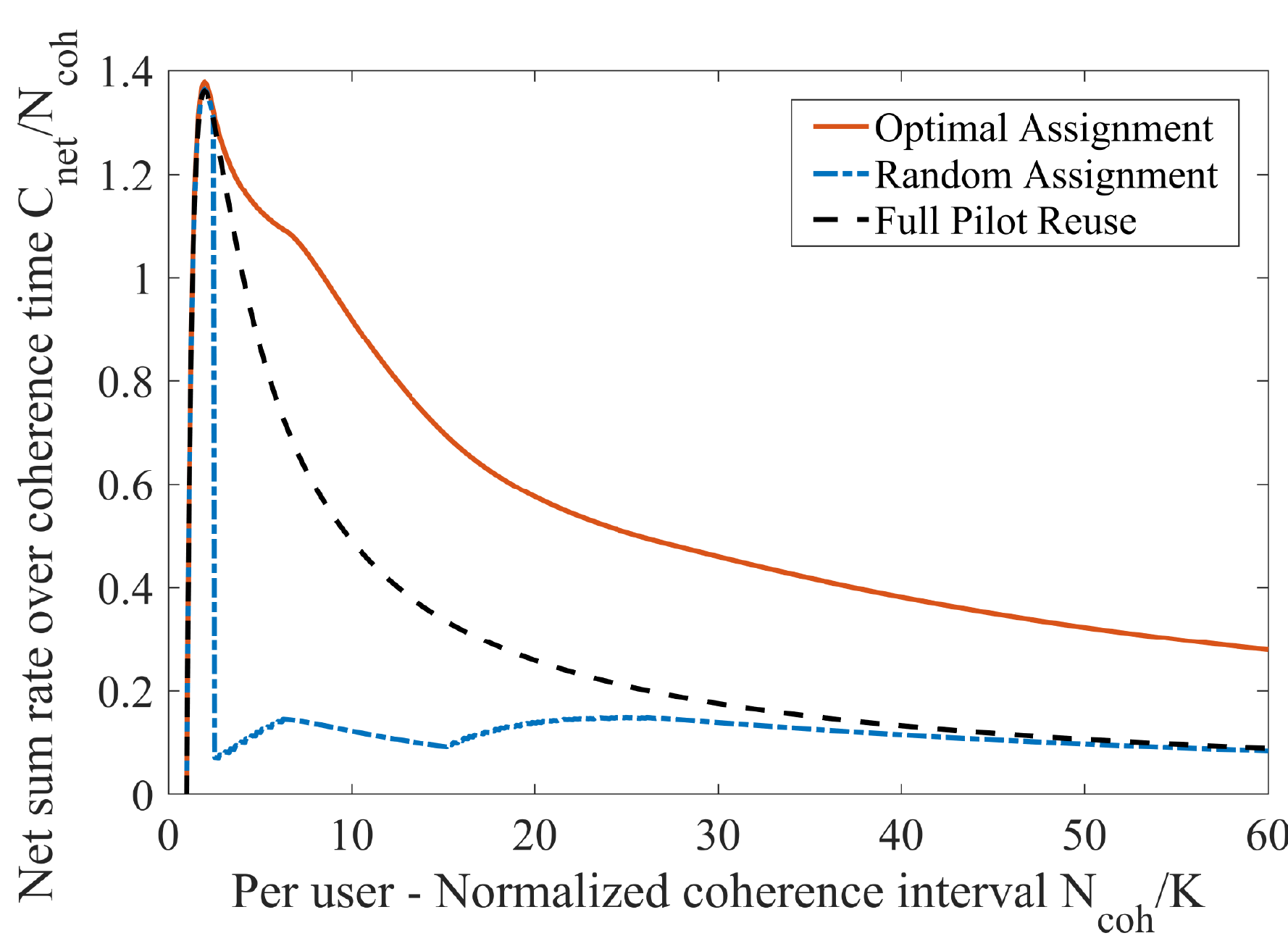}
	\caption{$C_{net}/N_{coh}$ versus $N_{coh}/K$ for different pilot assignment schemes}
	\label{Fig:K=14 Simulation2}
\end{figure}

Fig. \ref{Fig:K=14 Simulation2} shows $C_{net}/N_{coh}$ versus $N_{coh}/K$ for 
optimal, full reuse and random pilot assignment schemes. 
Again, $K=14$ is used to generate these plots, 
but the plots do not change noticeably 
for different values of $K$. 
These plots give an insight 
into how the net sum rate changes as $K$ decreases while $N_{coh}$ is held fixed.
Notice that the maximum net sum-rate occurs at $N_{coh}/K=2$ 
and at this point the optimal scheme reduces to full reuse,
consistent with the
Marzetta's analysis \cite{marzetta2010noncooperative}.
However, the message here is, again, that when we do not have a control over 
$N_{coh}/K=2$, the optimal pilot assignment strategy may have a substantial net sum-rate advantage 
over full pilot reuse.

To appreciate how large $N_{coh}/K$ can be in some real-world scenarios,
take an indoor office wireless channel with $T_{coh}=50$ micro-sec and $T_{del}=50$ nsec, yielding 
$N_{coh}=1000$. If the user density cannot be allowed to be more than $K=20$ users per cell, 
we would be focusing on $N_{coh}/K=50$ and at this point, the optimal scheme gives  a 300\% net sum-rate improvement over full pilot reuse. As another example, consider an urban outdoor environment with a fairly high user mobility giving rise to a 1 msec coherence time interval. With
a 2 micro-sec delay spread, for example, this gives $N_{coh}=500$, and 
assuming not more than $K=25$ users are to be served, we are interested in the net sum-rates at $N_{coh}/K=20$. From either Fig. \ref{Fig:K=14 Simulation} or Fig. \ref{Fig:K=14 Simulation2}, we see that a net throughput improvement of $121\%$ is possible via optimal pilot assignment, relative to full pilot reuse.

% Further Comments
\section{Further Comments}

\subsection{Performance analysis for finite antennas case}\label{Section:FiniteM}
The mathematical analysis conducted in this paper focuses on the scenario of base stations having infinitely many antennas. However, in practical applications of massive MIMO, finite BS antennas must be considered. For example, $128$ BS antennas have been used in the recent literature \cite{larsson2014massive, gao2015massive}. Thus, here we also consider optimal pilot assignment for finite $M$ cases.
For arbitrary $M$, the uplink spectral efficiency of massive MIMO systems with pilot reuse factor $\beta$ and maximal-ratio-combining (MRC) receiver is obtained in Theorem $1$ of  \cite{bjornson2016massive}. Our analysis on optimal pilot assignment for finite $M$ adopts modified versions of several notations defined in \cite{bjornson2016massive}. First, $\rho$ denotes the signal-to-noise ratio (SNR) value, and $L_{j}^{(i)}$ represents the set of cells which share pilot with cell $j$ by a reuse factor $3^i$. Moreover, based on $r_{jkl}$ defined in Section \ref{Section:SystemModel}, the following mathematical definitions are used:
\begin{gather}
\mu_{jl}^{(\omega)}\triangleq\mathbb{E}[(\frac{r_{lkl}}{r_{jkl}})^{\gamma \omega}] , \ \ \ \ \omega=1,2\\
\mu_0 \triangleq \sum_{l=1}^{L} \mu_{jl}^{(1)}, \ \ \ \ 
\mu_1^{(i)} \triangleq \sum_{l\in L_j^{(i)} \setminus \{j\}} \mu_{jl}^{(1)}, \ \ \ \ \nonumber\\
\mu_2^{(i)} \triangleq \sum_{l\in L_j^{(i)} \setminus \{j\}} (\mu_{jl}^{(1)})^2, \ \ \ \ 
\mu_3^{(i)} \triangleq \sum_{l\in L_j^{(i)} \setminus \{j\}} \mu_{jl}^{(2)}.
\end{gather}

Given the pilot assignment vector  $\mathbf{p}$, the interference term $I_{i}(M)$ for a user with pilot reuse factor $3^i$ can be expressed as 
\begin{align*}
I_{i} (M) &= \mu_3^{(i)} + \frac{\mu_3^{(i)}-\mu_2^{(i)}}{M} \nonumber\\
& \quad \quad + \frac{(K\mu_0 + \frac{1}{\rho})(1+\mu_1^{(i)}+\frac{1}{N_{pil}(\mathbf{p})\rho})}{M}.
\end{align*}
Here, $\mu_3^{(i)}$ represents the pilot contamination term remaining even if $M$ increases without bound, while the other terms are intra- and inter-cell interferences which are considered only when $M$ is finite.
The net throughput of a user with pilot reuse factor $3^i$ can be obtained as
$SE^{(i)} = \left( 1-\frac{N_{pil}(\mathbf{p})}{N_{coh}}\right)  \log_2\left( 1+\frac{1}{I_{i}(M)}\right),$
while the per-cell net throughput for a pilot assignment vector $\mathbf{p}=(p_{0},p_{1},\cdots,p_{\log_3 L-1})$ is
\begin{align}\label{Eqn:net_rate_finite_M}
C_{net}(\mathbf{p}, M)=\left( 1-\frac{N_{pil}(\mathbf{p})}{N_{coh}}\right)& \nonumber\\
\sum_{i=0}^{\log_3 L - 1} 3^{-i}& p_i \log_2\left( 1+\frac{1}{I_{i}(M)}\right),
\end{align}
which reduces to (\ref{net sum rate}) when $M$ tends to infinity. Then, the optimal pilot assignment vector for finite $M$ can be expressed as 
\begin{equation}
\mathbf{p}_{opt}(M) = \underset{\mathbf{p}\in P_{L,K}}{\arg\max}\ C_{net}(\mathbf{p},M).
\end{equation}

Here, we show the behavior of $\mathbf{p}_{opt}(M)$ for finite $M$, as observed by simulation results. Table 	\ref{Table:Optimal assignment(FiniteM)} shows the optimal pilot assignment vector for the $L=81, M=128, K=10$ case. For numerical calculation, the SNR value is set to $\rho = 5$dB while the setting for other parameters is specified in Section \ref{Section:NumResult}.
 As $N_{coh}/K$ increases, the optimal pilot assignment $\mathbf{p}_{opt}(M)$ chooses 
less aggressive pilot reuse, by utilizing a larger number $N_{pil}(\mathbf{p}_{opt}(M))$ of pilots. Moreover, the optimal assignment vector shows a similar pattern to what was observed when $M$ tended to infinity: 
tossing 1 from the left most non-zero component to increase the adjacent component by 3, for example $(9,3,0,0) \rightarrow (8,6,0,0)$. This implies that even in systems with finite BS antennas, where both pilot contamination and other interference terms are mixed, the optimal assignment chooses the way of reducing the number of users suffering the most severe pilot contamination problem. 

\begin{table}
	\small
	\caption{Optimal pilot assignment for finite antennas ($L=81,M=128,K=10$)}
	\centering
	\label{Table:Optimal assignment(FiniteM)}
	\begin{tabular}{|c|c|c|}
		\hline
		$N_{coh}/K$ & $\mathbf{p}_{opt}(M)$  & $N_{pil}(\mathbf{p}_{opt}(M))$\tabularnewline
		\hline
		$0\sim 4.5$ & $(10,0,0,0)$ & $10$\tabularnewline
		$4.5\sim 4.9$ & $(9,3,0,0)$ & $12$\tabularnewline
		$4.9\sim 5.3$ & $(8,6,0,0)$ & $14$\tabularnewline
		$5.3\sim 5.7$ & $(7,9,0,0)$ & $16$\tabularnewline
		$5.7\sim 6.1$ & $(6,12,0,0)$ & $18$\tabularnewline
		$\vdots$ & $\vdots$ & $\vdots$\tabularnewline
		\hline
	\end{tabular}
\end{table}

Fig. \ref{Fig:FiniteM_Simulation} illustrates the per-user net rate for various pilot assignments as a function of $M$, under the setting of $L=27, K=10, N_{coh}=200$. Here, the conventional assignment represents the full pilot reuse case.
%, and the random assignment utilizes the max-min power control algorithm suggested in \cite{ngo2016cell}. 
The solid lines represent the simulation results for finite $M$, while the asymptotic values for $M=\infty$ are illustrated in dashed/dash-dot lines. Here, the performance gap between optimal assignment and conventional assignment can be observed, at the number of BS antennas as small as $10$. In the case of $M=128$ and $M=1024$, the optimal assignment outperforms conventional assignment by $40\%$ and $84\%$, respectively. Note that the performance of conventional assignment saturates around $M=128$, so that additional BS antennas are redundant, while the optimal assignment has sufficient margin to increase throughput by increasing $M>128$. Therefore, in mm-Wave operation which can allow more uncorrelated antennas in a given area \cite{swindlehurst2014millimeter}, optimal assignment can enjoy a larger performance gain compared to conventional full pilot reuse. 
%Moreover, the optimal assignment outperforms the random assignment with power control of \cite{ngo2016cell}, for the large $M$ regime, which can be explained as follows.  
%The power control algorithm of \cite{ngo2016cell} is designed to maximize the minimum user rate, which helps to avoid severe pilot contamination to all users. However, the max-min power control does not guarantee an increased average throughput, $C_{net}/K$.
%\textcolor{red}{Moreover, the optimal assignment outperforms the random assignment with power control of \cite{ngo2016cell}, for the large $M$ regime, which can be explained as follows. The power control algorithm of \cite{ngo2016cell} is designed to maximize the minimum user rate, which helps to avoid severe pilot contamination to all users. However, the max-min power control does not guarantee an increased average throughput, $C_{net}/K$.}

%\begin{figure}[!t]
%	\centering
%	\includegraphics[height=45mm]{N400_K20_rev_rev.pdf}
%	\caption{Net rate versus $M$ for optimal/conventional pilot assignments}
%	\label{Fig:FiniteM_Simulation}
%\end{figure}

%\begin{figure}
%	\centering
%	\subfloat[][$K=10, N_{coh}=200$]{\includegraphics[height=53mm]{[Netrate_Compare]L27K10N200rev.eps}\label{Fig:FiniteM_Simulation}}
%	\quad \quad
%	\subfloat[][Optimal Assignment, $N_{coh} = 2000$]{\includegraphics[height=53mm]{OptCompareL27N2000_rev.eps}\label{Fig:VariousM/K_Simulation}}
%	\caption{Net rate versus $M$ for $L=27$}
%\end{figure}

\begin{figure}[!t]
	\centering
	\includegraphics[height=50mm]{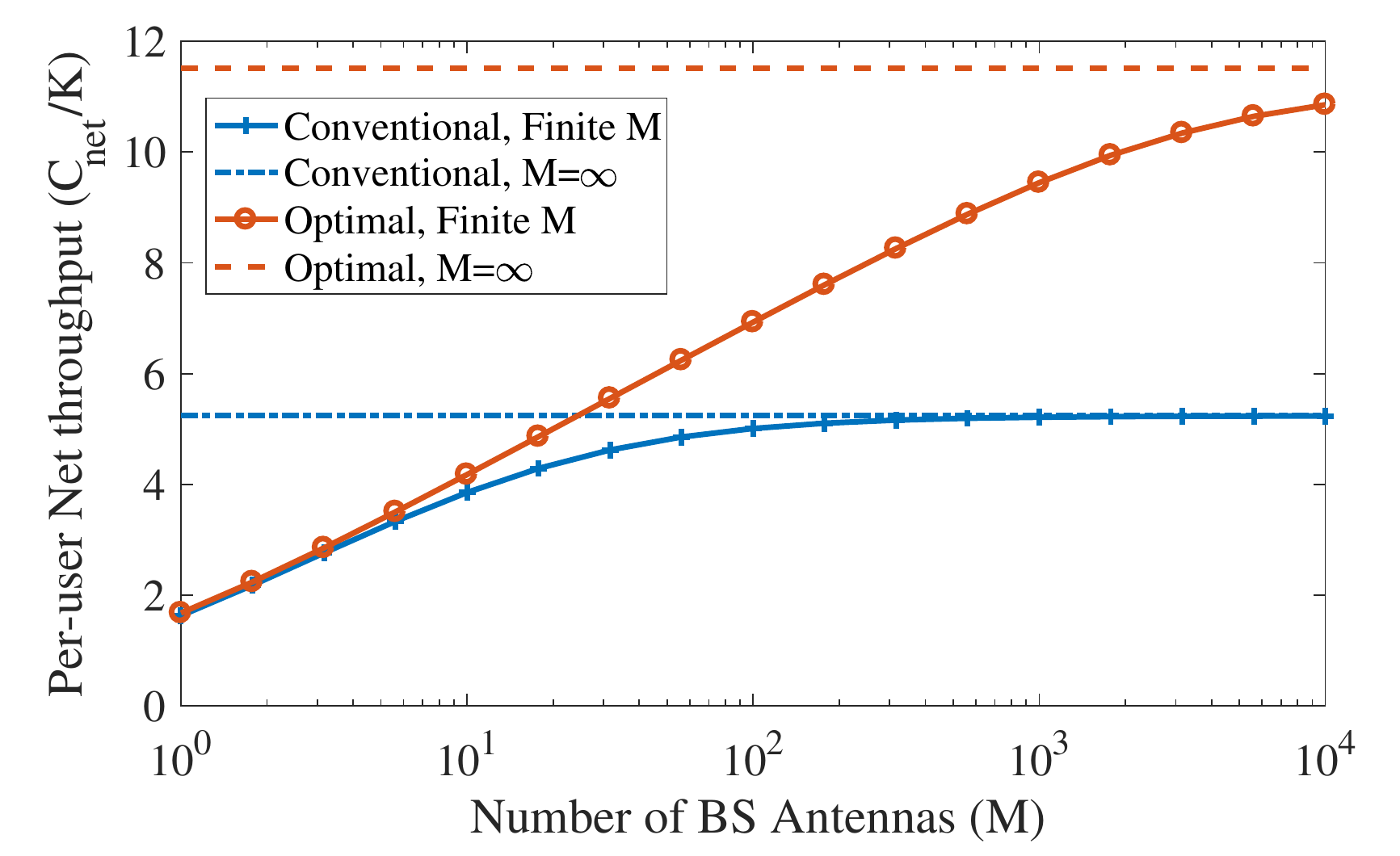}
	\caption{Net rate versus $M$ for various pilot assignments ($L=27, K=10, N_{coh} = 200 $)}
	\label{Fig:FiniteM_Simulation}
\end{figure}

To provide an insight on the performance of the optimal assignment, we also plotted the per-user throughput for various ratios of $M/K=2,4,10,20$. As can be seen in Fig. \ref{Fig:VariousM/K_Simulation}, the throughput keeps increasing as a function of $M$, as long as $M/K$ is large enough (greater than $10$ in the plots provided). However, for smaller ratios ($M/K<10$) the net throughput peaks and then starts to decreases as $M$ increases. We observe that this is because as $M$ increases, the decreasing rate of the coefficient ($1-\frac{N_{pil} (p)}{N_{coh}}$ ) in (\ref{Eqn:net_rate_finite_M}) is higher than the increasing rate of the main summation term when $M/K$ is not sufficiently high (i.e., massive MIMO effect is not realized).
%\textcolor{red}{To provide an insight on the performance of the optimal assignment, we also plotted the per-user throughput for various ratios of $M/K=2,4,10,20$. As can be seen in Fig. \ref{Fig:VariousM/K_Simulation}, the throughput keeps increasing as a function of $M$, as long as $M/K$ is large enough (greater than $10$ in the plots provided). However, for smaller ratios ($M/K<10$) the net throughput peaks and then starts to decreases as $M$ increases. We observe that this is because as $M$ increases, the decreasing rate of the coefficient ($1-\frac{N_{pil} (p)}{N_{coh}}$ ) in (\ref{Eqn:net_rate_finite_M}) is higher than the increasing rate of the main summation term when $M/K$ is not sufficiently high (i.e., massive MIMO effect is not realized). }

\begin{figure}[!t]
	\centering
	\includegraphics[width=85mm]{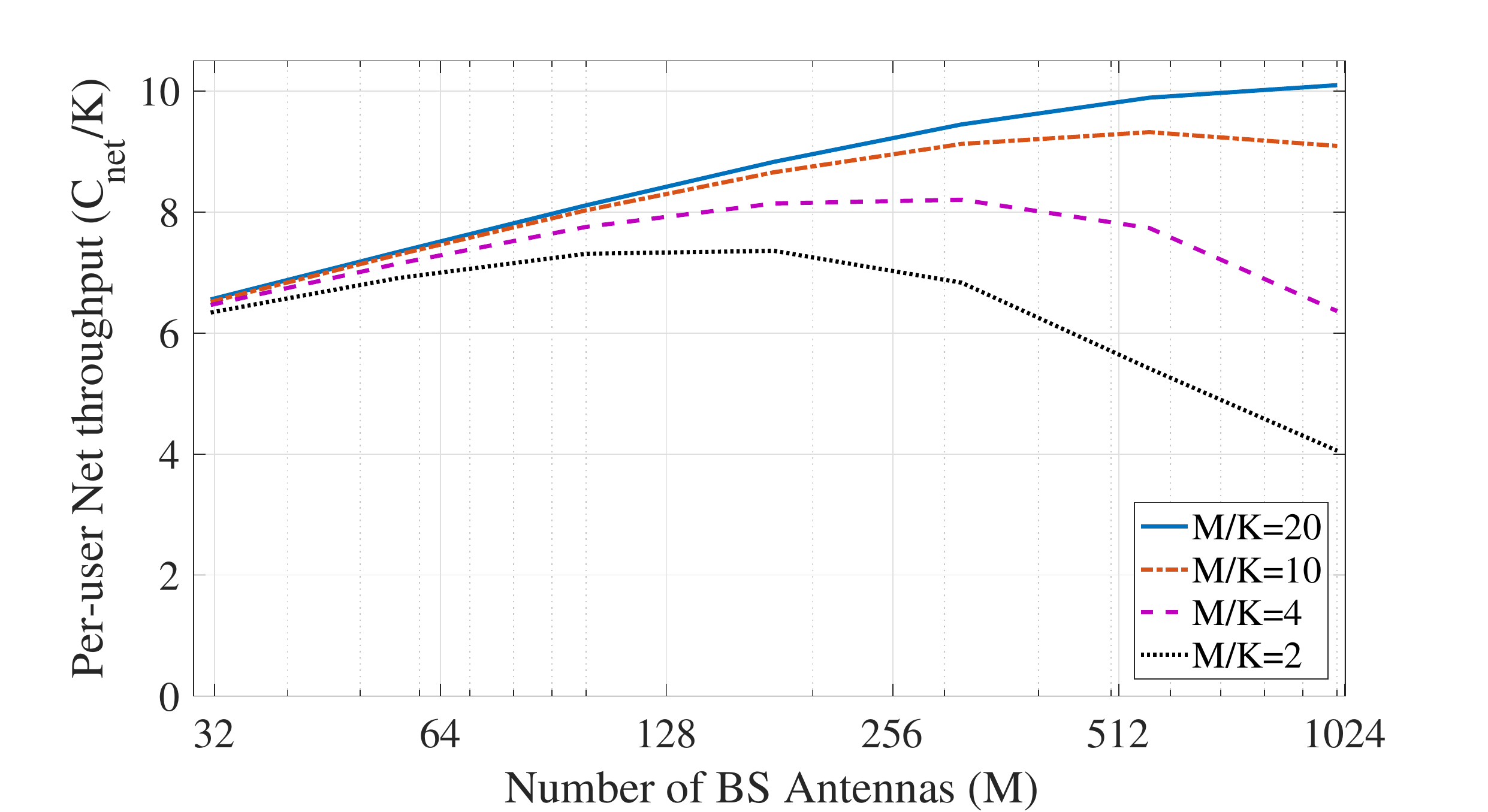}
	\caption{Net rate versus $M$ of the optimal assignment, for various $M/K$ settings ($L=27, N_{coh} = 2000$)}
	\label{Fig:VariousM/K_Simulation}
\end{figure}

We also observed the distribution of the achievable rate for each user. In Fig. \ref{Fig:CDF}, the cumulative distribution function (CDF) curves of three pilot assignment rules are illustrated: the optimal assignment, conventional full pilot reuse, and the random assignment.
%s discussed in \cite{ngo2016cell}. 
%The authors of \cite{ngo2016cell} considered two power control algorithms in random assignments: one is the uniform transmit power and the other is max-min power control which maximizes the minimum user rate. 
Both Fig.\ref{Fig:CDF_K1} and Fig.\ref{Fig:CDF_K40} show that the proposed optimal assignment helps all users to enjoy high achievable rates, compared to other pilot assignment rules. This is because the optimal assignment allocates additional pilots so as to relieve the most severely affected users first. %While the max-min power control of \cite{ngo2016cell} increases the minimum user rate (note that the minimum rate of the purple dashed line is higher than the minimum rate of the blue dash-dotted line), the optimal assignment is far more efficient because the latter increases both the average rate and the minimum rate. 

%\textcolor{red}{We also observed the distribution of the achievable rate for each user. In Fig. \ref{Fig:CDF}, the cumulative distribution function (CDF) curves of four pilot assignment rules are illustrated: the optimal assignment, conventional full pilot reuse, and two random assignments discussed in \cite{ngo2016cell}. The authors of \cite{ngo2016cell} considered two power control algorithms in random assignments: one is the uniform transmit power and the other is max-min power control which maximizes the minimum user rate. Both Fig.\ref{Fig:CDF_K1} and Fig.\ref{Fig:CDF_K40} show that the proposed optimal assignment helps all users to enjoy high achievable rates, compared to other pilot assignment rules discussed in previous works. This is because the optimal assignment allocates additional pilots so as to relieve the most severely affected users first. While the max-min power control of \cite{ngo2016cell} increases the minimum user rate (note that the minimum rate of the purple dashed line is higher than the minimum rate of the blue dash-dotted line), the optimal assignment is far more efficient because the latter increases both the average rate and the minimum rate.  }

%\begin{figure}
%	\centering
%	\subfloat[][K=1]{\includegraphics[height=53mm]{[Cdf_Compare]L27K1M100_rev_rev.eps}\label{Fig:CDF_K1}}
%	\quad \quad
%	\subfloat[][K=40]{\includegraphics[height=53mm]{[Cdf_Compare]L27K40M100_rev.eps}\label{Fig:CDF_K40}}
%	\caption{Cumulative distribution function of per-user achievable rate for various pilot assignments (L=27, M=100)}
%	\label{Fig:CDF}
%\end{figure}

\begin{figure}
	\centering
	\subfloat[][K=1]{\includegraphics[height=45mm]{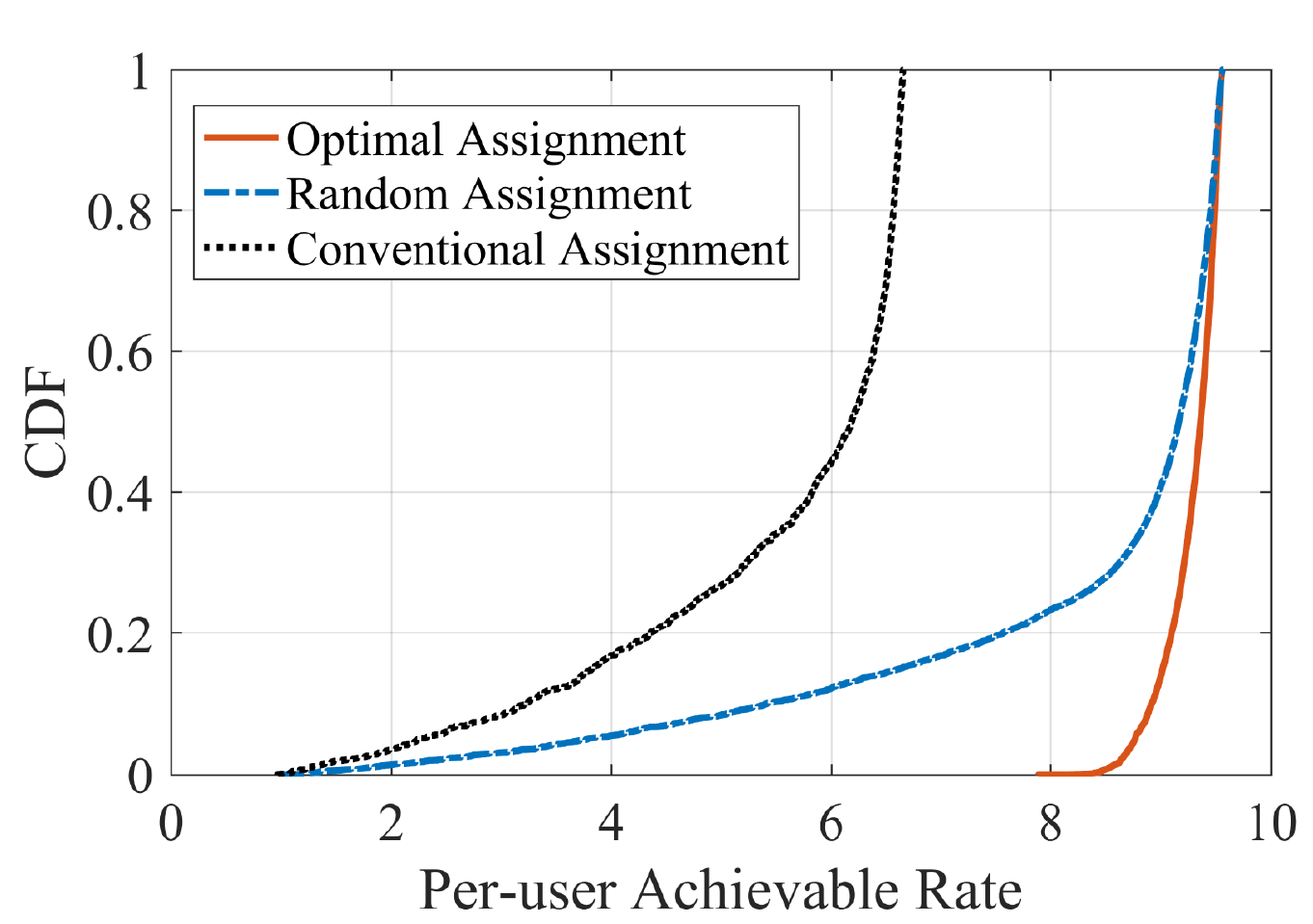}\label{Fig:CDF_K1}}
	\quad \quad
	\subfloat[][K=40]{\includegraphics[height=50mm]{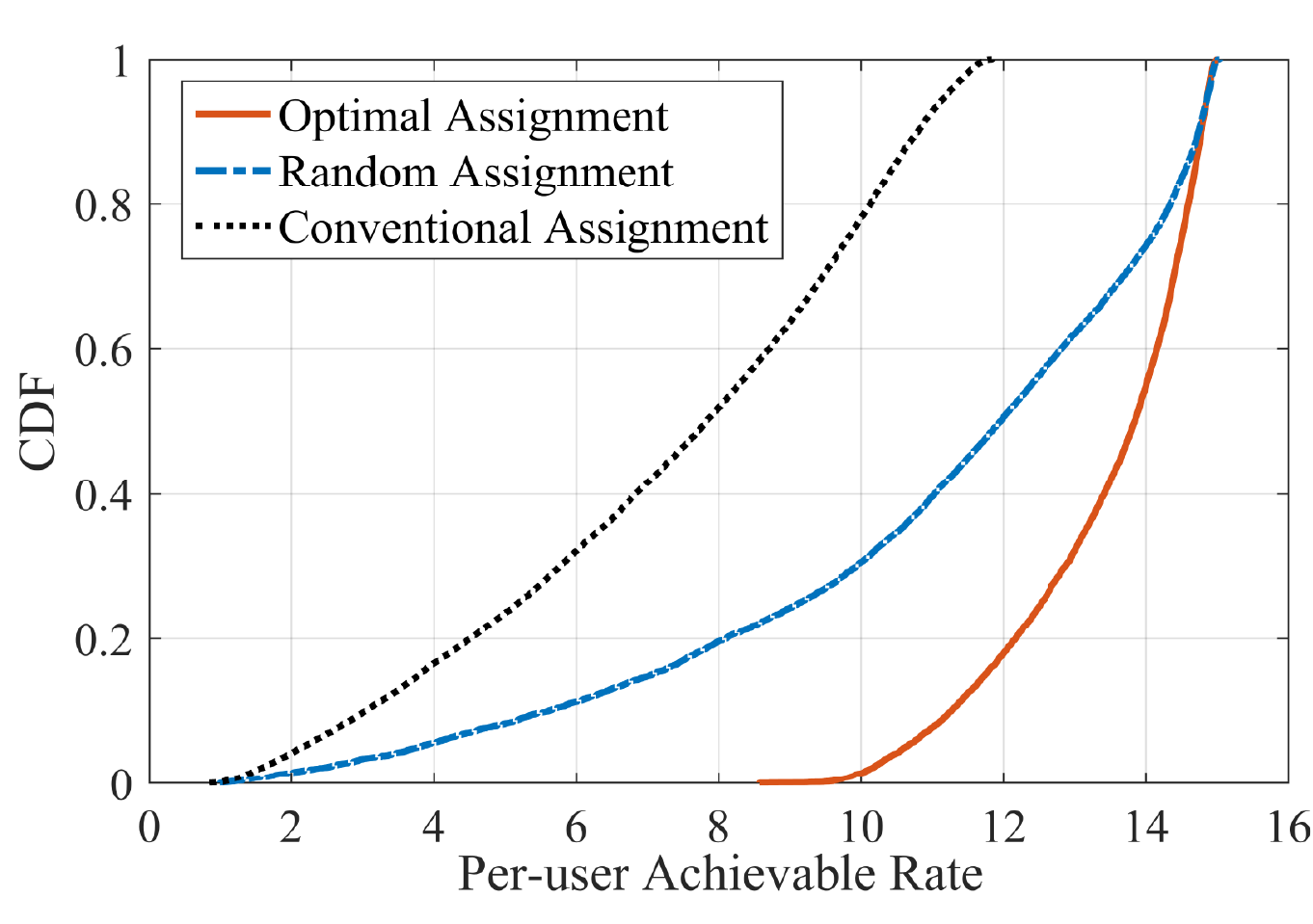}\label{Fig:CDF_K40}}
	\caption{Cumulative distribution function of per-user achievable rate for various pilot assignments (L=27, M=100)}
	\label{Fig:CDF}
\end{figure}

%\begin{figure}[!t]
%	\centering
%	\includegraphics[height=55mm]{[Cdf_Compare]L27K1M100.eps}
%	\caption{Cumulative distribution function of per-user achievable rate for various pilot assignments (L=27, K=1, M=100)}
%\end{figure}
%
%
%\begin{figure}[!t]
%	\centering
%	\includegraphics[height=55mm]{[Cdf_Compare]L27K40M100.eps}
%	\caption{Cumulative distribution function of per-user achievable rate for various pilot assignments (L=27, K=40, M=100)}
%\end{figure}

\subsection{Optimal Training Time Portion of the Coherent Time Interval}

An interesting remaining question is: what should be the optimal portion of the time allocated for pilot (and thus for training) in massive MIMO with interfering cells as $N_{coh}$ increases? 
Fig. \ref{Fig:Portion} illustrate how much training is used for optimal assignment in the cases of $K=1$ and $K=14$. For a given $K$, the ratio $\frac{N_{pil}(\mathbf{p}_{opt})}{N_{coh}}$ is calculated for various $N_{coh}$ values. As $N_{coh}/K$ increases, both plots in Fig. \ref{Fig:Portion} show a non-vanishing portion of the coherence time used for the pilots in optimal assignment, while the pilot portion shrinks in full pilot reuse. For a given $K$, the curve for optimal assignment consists of a family of curves corresponding to $\mathbf{p}'_{opt}(N_{p0})$ for $N_{p0}=K, K+2, K+4, \cdots, LK/3$. Since the optimal $\mathbf{p}_{opt}$ changes as $N_{coh}/K$ increases, we have a discontinuous function made up of a family of exponentially decaying functions. 

%\begin{figure}[!t]
%\centering
%        \includegraphics[height=50mm]{K=1_161017.pdf}
%    \caption{Optimal portion of the coherent time dedicated for training ($L=81,K=1$)}
%    \label{Fig:ratio_of_pilots_K=1}
%    
%\end{figure}
%
%\begin{figure}[!t]
%\centering
%        \includegraphics[height=50mm]{K=14_161017.pdf}
%    \caption{Optimal portion of the coherent time dedicated for training ($L=81,K=14$)}
%    \label{Fig:ratio_of_pilots_K=14}
%    
%\end{figure}

\begin{figure}
	\centering
	\subfloat[][K=1]{\includegraphics[width=0.38\textwidth]{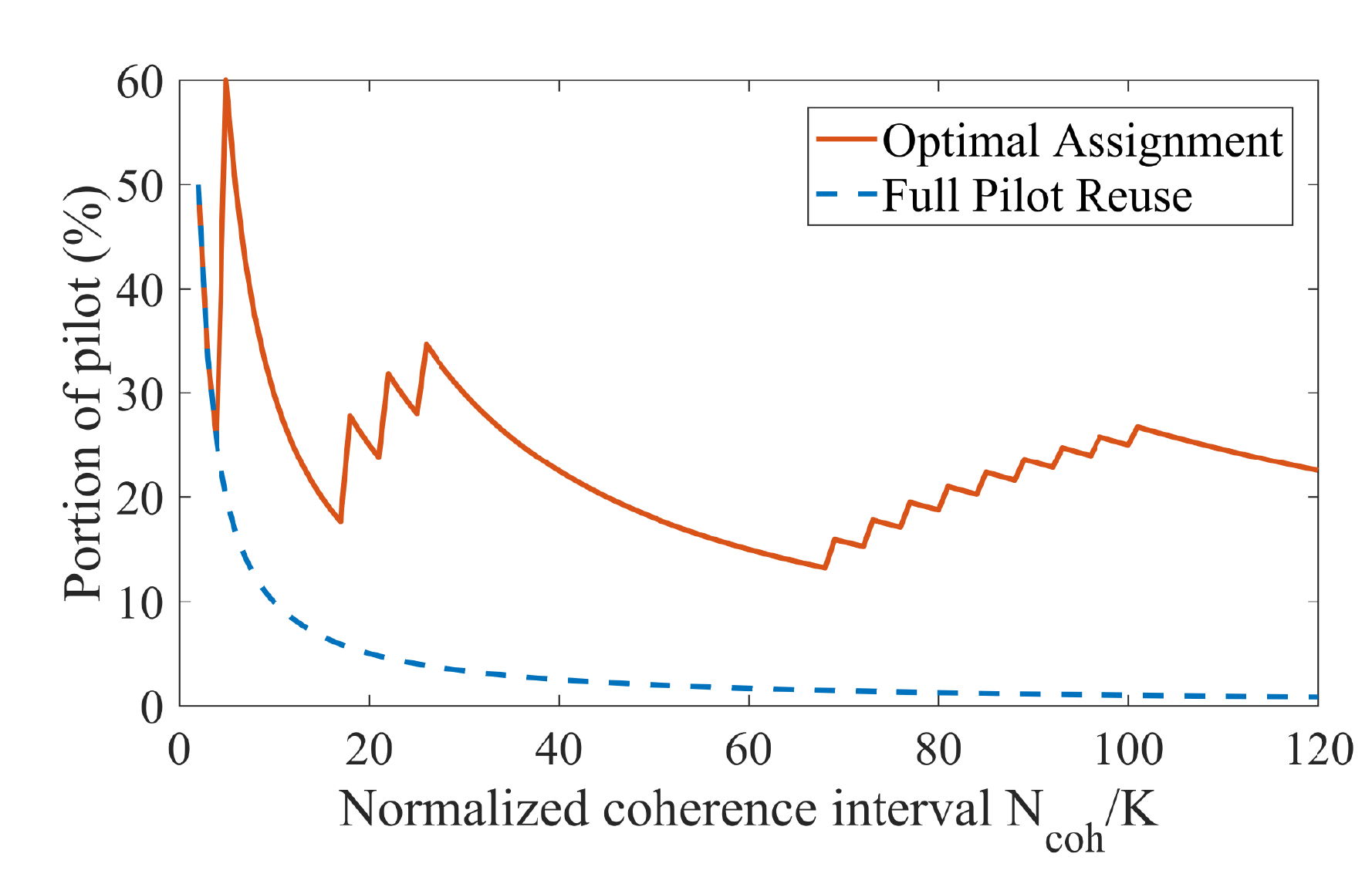}\label{Fig:Portion1}}
	\quad \quad
	\subfloat[][K=14]{\includegraphics[width=0.38\textwidth]{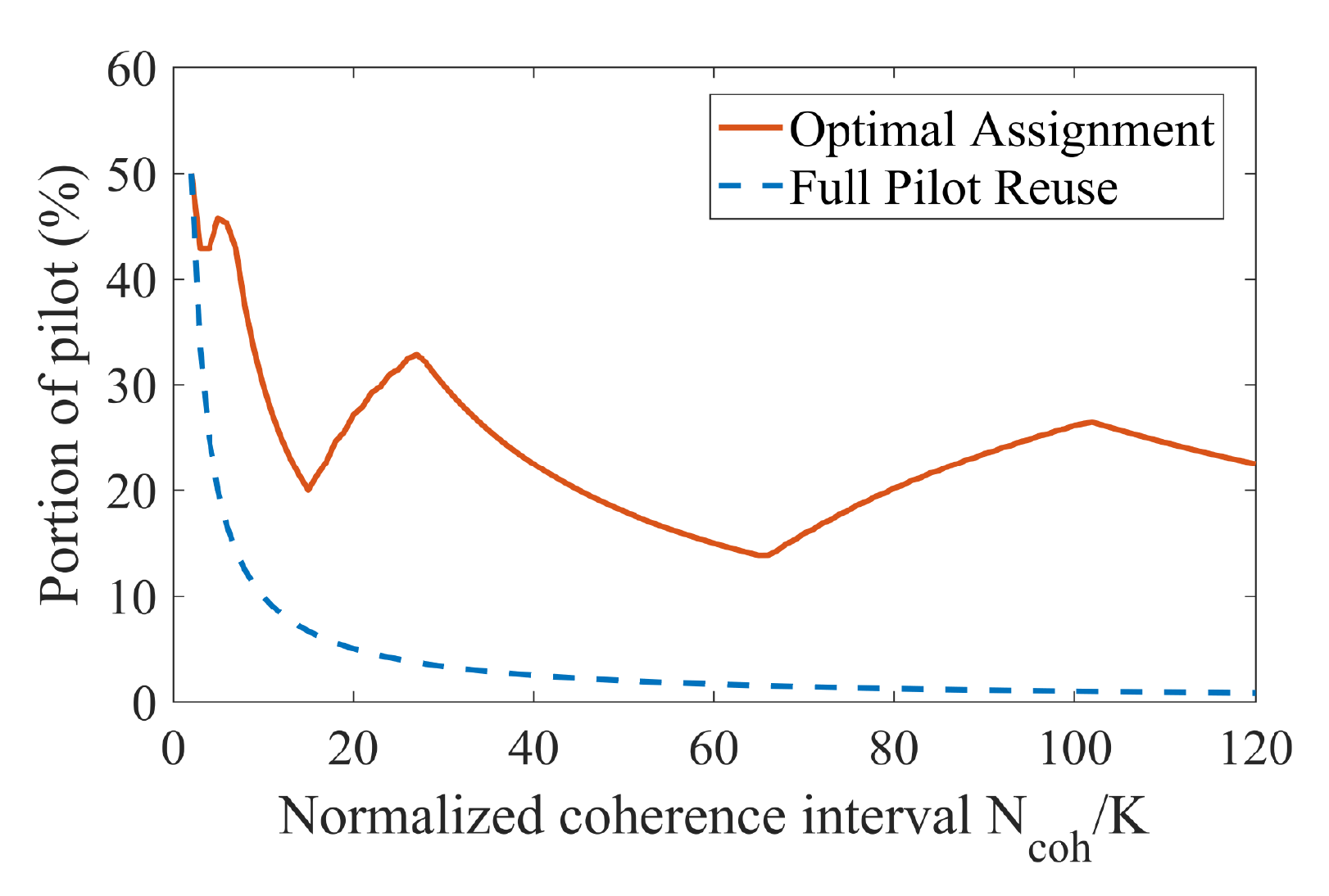}\label{Fig:Portion14}}
	\caption{Optimal portion of the coherent time dedicated for training ($L=81$)}
	\label{Fig:Portion}
\end{figure}

\subsection{Comparison with Cell Partitioning in Frequency Reuse}

The analysis in this paper is based on hierarchical partitioning originated from 3-way partitioning, as described in section III-A. However, this 3-way partitioning is identical to the partitioning used in frequency reuse factor of three. In Appendix of \cite{mac1979advanced}, fundamentals of hexagonal cellular geometry are discussed. In hexagonal cell systems, the number of cells per cluster for frequency reuse is given by the form $N=i^2+ij+j^2$, where $i$ and $j$ are integers which shape the geometry of the clusters. For example, $i=1, j=1$ result in $N=3$, corresponding to a frequency reuse factor of three. 

The hierarchical partitioning considered in this paper can be described by the geometrical model in \cite{mac1979advanced} to a certain extent. Consider a tree structure having leaves with same depth. This tree structure corresponds to pilot assignment vectors with only one nonzero element like $(1,0,\cdots,0)$, $(0,3,0,\cdots,0)$, $(0,0,9,0,\cdots,0)$, $\cdots$. Each assignment method is identical to the cell geometry of frequency reuse factor of $1,3,9,\cdots$ which can be generated by appropriate $i$ and $j$ (i.e., $1=1^2+1\cdot0+0^2$, $3=1^2+1\cdot1+1^2$, $9=3^2+3\cdot0+0^2$, $\cdots$). However, other pilot assignment vectors considered here cannot be expressed by the form of \cite{mac1979advanced}. For example, pilot assignment in Fig. \ref{Fig:coloring example} is generated by a mixture of frequency reuse factors of $3$ and $9$, which yield leaf nodes at depth $1$ and $2$.

Finally, we note that the suggested pilot assignment strategy can obviously be utilized in conjunction with the existing frequency reuse scheme. 

\subsection{Optimization problem for maximizing net weighted sum-rate}

In the example discussed towards the end of Section \ref{Multiuser_numerical result}, for $N_{coh}=500$, the maximum net sum-rate is achieved when there are 
$K=500/2=250$ users. Using the plots in Fig. \ref{Fig:K=14 Simulation2}, the corresponding normalized net sum-rate is
$C_{net}/N_{coh}=3.2$, which translates to a per-user net rate of $(3.2 \times 500)/250 = 6.4$. 
In comparison, for the case of the fixed number of $K=25$ users, the full-pilot reuse gives $C_{net}/N_{coh}=0.6$
or a per-user net rate of $(0.6 \times 500)/25 = 12$ (and using the optimal pilot assignment this rate improves to 16.2). This argument shows that while keeping the ratio $N_{coh}/K$ to 2 would maximize the net sum-rate, it may degrade individual user rates to a level that would be 
highly undesirable in many real-world scenarios. In practical applications, maximizing the sum-rate may not be an ideal strategy; often it would make sense to maximize a weighted sum-rate (WSR). Finding general pilot assignment solutions that maximize WSR is an interesting research direction, which we postpone to next paper.

%While finding general pilot assignment solutions that maximize WSR is an interesting research direction, here we only take a glimpse into this issue by considering some simple example cases.

\subsection{Impact with regards to ultra-densification}
According to current literatures \cite{gotsis2016ultradense, bhushan2014network,ge20165g}, one of the main directions for 5G is ultra-dense network, where density of BSs in given area gets higher in order to support data demand from massive devices. Specifically, heterogeneous network (HetNet) with both macro-BS and micro-BS are considered in \cite{sanguinetti2015interference}. Here, a single macro-BS (with massive antennas) at the center of each cell communicates with high-mobility UTs throughout the cell, while several micro-BSs (with single antenna each) spread in the cell supports low-mobility UTs within a local area. Macro-BS and micro-BSs are connected by wireless links, by allocating some antennas in macro-BS for backhaul/fronthauling. 
In this HetNet scenario with multiple interfering macro-cells, suggested optimal pilot assignment can be applied to increase the throughput of high-mobility users supported by macro-BSs.

% Conclusion
\section{Conclusion}
In a massive MIMO system with interfering cells, allowing neighboring cells to use different
sets of pilot sequences can effectively mitigate the pilot contamination problem and increase
the achievable net throughput of the system, when the appropriate pilot assignment strategy is
applied. Assuming hexagonal cells and equi-distance hierarchical partitioning, an optimal pilot
assignment strategy has been identified that gives substantial throughput advantages relative to
random pilot assignment or full pilot reuse when the given coherence time interval $N_{coh}$ and
the number of users $K$ has sufficiently large ratio. As $N_{coh}/K$ increases, the optimal number
of pilots also grows, where the additional pilots are allocated so as to relieve
the most severely affected users first. Finally, we add that
it would be interesting to further explore pilot assignment strategies when the objective is not
about maximizing the sum rate but rather on guaranteeing some minimal performance level to
all users or maximizing a weighted sum rate to prioritize the services.

% or


% appendices
\appendices
% Appendix A
\section{Proofs of Lemmas 1 and 2}

\begin{figure}[!t]
	\centering
	\includegraphics[height=20mm]{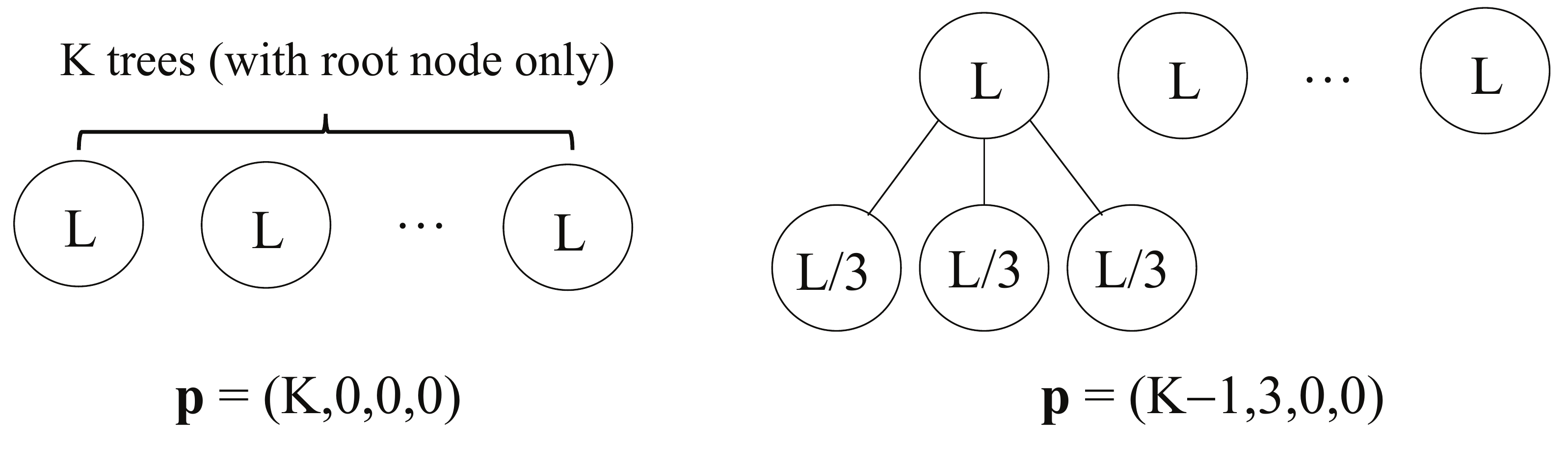}
	\caption{Tree structure for $K$ users}
	\label{Lemma1_proof_1}
\end{figure}

% Proof of Lemma 1
\subsection{Lemma 1}
\begin{proof}[Proof:\nopunct]
In order to support $K$ users in each cell, valid pilot assignment starts with full pilot reuse which is illustrated in left side of Fig. \ref{Lemma1_proof_1}. Since $N_{pil}(\mathbf{p})$ is the number of leaves, this assignment has $N_{pil}(\mathbf{p})=K$. When 3-way partitioning is applied to a tree (as on the right side of Fig. \ref{Lemma1_proof_1}), the number of leaf nodes $N_{pil}(\mathbf{p})$ increases by 2 (reducing 1 leaf node denoted as $L$, and increasing 3 leaf nodes denoted as $L/3$). Similarly, consecutive 3-way partitioning increases $N_{pil}(\mathbf{p})$ by 2, until all $K$ trees reaches the maximum depth, $log_3 L -1$. Therefore, the maximum value of $N_{pil}(\mathbf{p})$ is $LK/3$, which completes the proof.
\end{proof}

% Proof of Lemma 2
\subsection{Lemma 2}
\begin{proof}[Proof:\nopunct]
Denote $I_{0} \triangleq log_3 L -1 $.
Consider arbitrary $\mathbf{p}=(p_{0}, p_{1}, \cdots, p_{I_{0}}) \in \Omega(N_{p0})$ and its corresponding transition vector $\mathbf{t}=(t_{0}, t_{1}, \cdots, t_{I_{0} - 1})$ defined by (\ref{ptot}).
Then, using the relationship given in (\ref{ptot}), we can obtain a closed-form expression for $t_{i}$ for $1\leq i \leq  I_{0}-1$ as follows:
%$t_{i} = -p_{i} + 3t_{i-1} = -p_{i} + 3(-p_{i-1} + 3t_{i-2})= -p_{i} + 3\left(-p_{i-1} + 3(-p_{i-2}+3t_{i-3})\right)= \cdots= -\displaystyle\sum_{s=0}^{i-1}3^{s}p_{i-s} + 3^{i}t_{0}= -\displaystyle\sum_{s=0}^{i-1}3^{s}p_{i-s} + 3^{i}(K-p_{0})= -\displaystyle\sum_{s=0}^{i}3^{s}p_{i-s} + K3^{i}= 3^{i}(K-p_{0}-\frac{p_{1}}{3}-\cdots-\frac{p_{i}}{3^{i}})$.
\begin{align*}
t_{i} &= -p_{i} + 3t_{i-1} = -p_{i} + 3(-p_{i-1} + 3t_{i-2}) \\
&= -p_{i} + 3\left(-p_{i-1} + 3(-p_{i-2}+3t_{i-3})\right)= \cdots\nonumber\\
&= -\displaystyle\sum_{s=0}^{i-1}3^{s}p_{i-s} + 3^{i}t_{0}= -\displaystyle\sum_{s=0}^{i-1}3^{s}p_{i-s} + 3^{i}(K-p_{0}) \nonumber\\
&= -\displaystyle\sum_{s=0}^{i}3^{s}p_{i-s} + K3^{i} = 3^{i}(K-p_{0}-\frac{p_{1}}{3}-\cdots-\frac{p_{i}}{3^{i}}).
\end{align*}
The upper bound and lower bound for $t_{i}$ (for $1\leq i \leq  I_{0}-1$) can be obtained as:
\begin{align*}
t_{i} &= 3^{i}(K-p_{0}-\frac{p_{1}}{3}-\cdots-\frac{p_{i}}{3^{i}}) \leq K3^{i},
\end{align*}
since $p_{s} \geq 0 \enspace \forall s \in \{0, 1, \cdots, I_{0}\}$ by (\ref{pilot assgning vector}).
Here, the equality holds iff $p_{s} = 0 \enspace \forall s \in \{0, 1, \cdots, i\}$. 
Similarly, using (\ref{pilot assgning vector}), 
\begin{align*}
t_{i} &= 3^{i}(K-p_{0}-\frac{p_{1}}{3}-\cdots-\frac{p_{i}}{3^{i}}) \nonumber\\
&= 3^{i}(\displaystyle\sum_{s=0}^{I_{0}}p_{s}3^{-s} - \displaystyle\sum_{s=0}^{i}p_{s}3^{-s}) = 3^{i}\displaystyle\sum_{s=i+1}^{I_{0}}p_{s}3^{-s} \geq 0
\end{align*}
where equality holds iff $p_{s} = 0 \enspace \forall s \in \{i+1, i+2, \cdots, I_{0}\}$.
Therefore, for $1\leq i \leq  I_{0}-1$, $0 \leq t_{i} \leq K3^{i}$ holds. However, $0 \leq t_{0}=K-p_{0} \leq K=K3^{0}$ since $0 \leq p_{0} \leq K$ by (\ref{pilot assgning vector}). Moreover, using the relationship in (\ref{ptot}), we can confirm $t_{i}$ are integers, since $p_{i}$ are integers by (\ref{pilot assgning vector}). The overall result can be combined as
\begin{equation}\label{proof_4_2_equation2}
t_{i} \in  \{0, 1, 2, \cdots, K3^{i}\} \qquad \forall i \in \{0, 1, \cdots, I_{0}-1 \}.
\end{equation}
On the other hand, using (\ref{ptot}), (\ref{ttop}) and $\mathbf{p} \in \Omega(N_{p0})$,
\begin{align*}
\displaystyle\sum_{i=0}^{I_{0}-1}t_{i} &= t_{0} + \displaystyle\sum_{i=1}^{I_{0}-1}t_{i} = (K-p_{0}) + \displaystyle\sum_{i=1}^{I_{0}-1}(-p_{i} + 3t_{i-1}) \nonumber\\
&= K - \displaystyle\sum_{i=0}^{I_{0}-1}p_{i} + 3\displaystyle\sum_{i=0}^{I_{0}-2}t_{i}\nonumber\\
&= K - \displaystyle\sum_{i=0}^{I_{0}}p_{i} + 3\displaystyle\sum_{i=0}^{I_{0}-1}t_{i} = K - N_{p0} + 3\displaystyle\sum_{i=0}^{I_{0}-1}t_{i},
\end{align*}
which is to say
\begin{equation}\label{proof_4_2_equation3}
\displaystyle\sum_{i=0}^{I_{0}-1}t_{i}=\frac{N_{p0}-K}{2}.
\end{equation}
(\ref{proof_4_2_equation2}) and (\ref{proof_4_2_equation3}) complete the proof. 
\end{proof}

%\section{}

%\section*{Acknowledgment}


\ifCLASSOPTIONcaptionsoff
  \newpage
\fi



%\bibliographystyle{IEEEtran}

\bibliographystyle{IEEEtran}
\bibliography{IEEEabrv,WC2016}
%\begin{thebibliography}{1}
%
%\bibitem{Ngo2016Downlink}
%Ngo HQ, Larsson EG. No Downlink Pilots are Needed in Massive MIMO. arXiv preprint arXiv:1606.02348. 2016 Jun 7. 
%
%\end{thebibliography}


\begin{IEEEbiography}[{\includegraphics[width=1in,height=1.25in,clip,keepaspectratio]{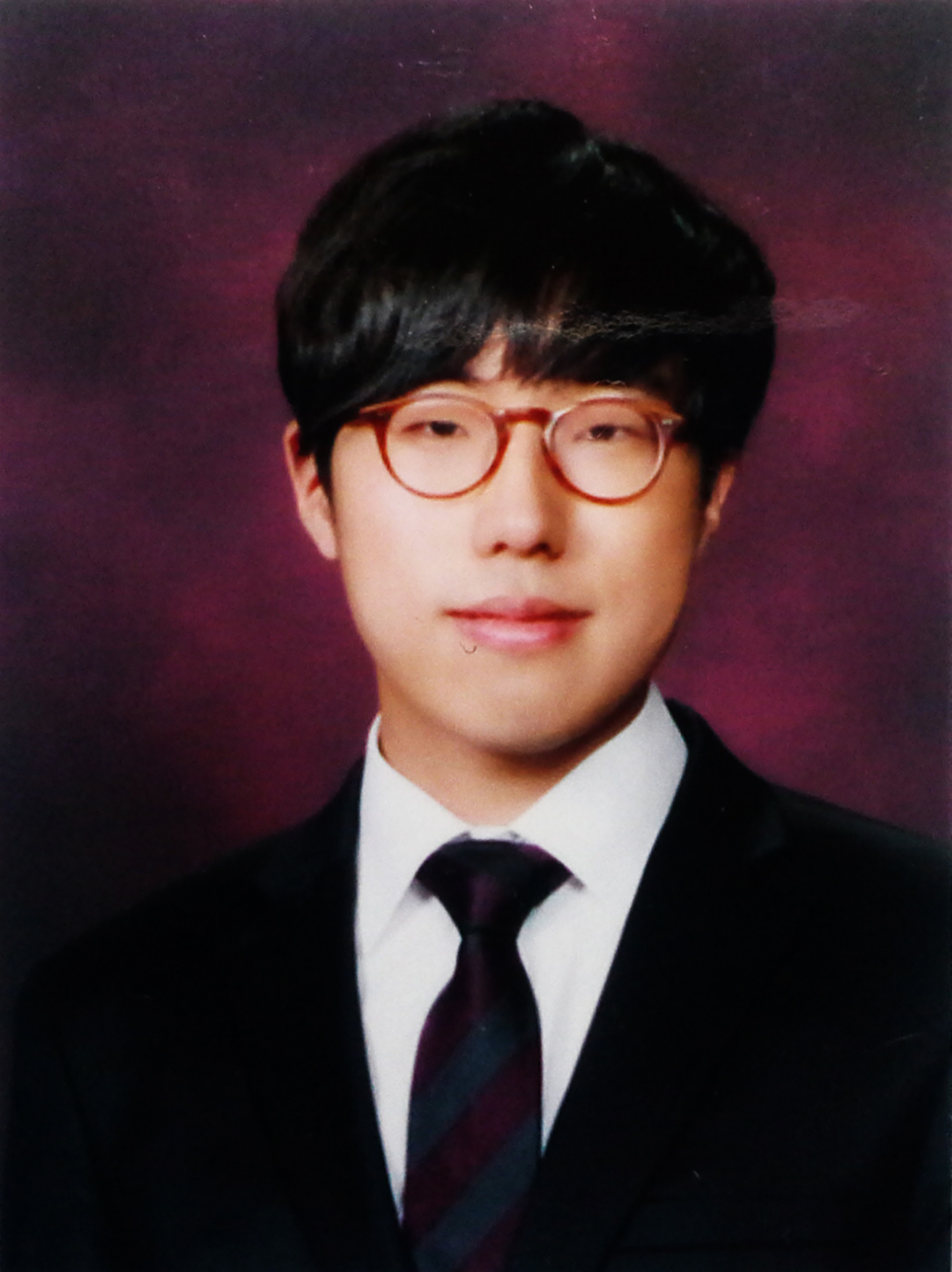}}]{Jy-yong Sohn}
	(S'15) received the B.S. and M.S. degrees in electrical engineering from the Korea Advanced Institute of Science and Technology (KAIST), Daejeon, Korea, in 2014 and 2016, respectively. He is currently pursuing the Ph.D. degree in KAIST. His research interests include massive MIMO effects on wireless multi cellular system and 5G Communications, with a current focus on distributed storage and network coding. He is a recipient of the IEEE international conference on
	communications (ICC) best paper award in 2017. 
\end{IEEEbiography}
\begin{IEEEbiography}[{\includegraphics[width=1in,height=1.25in,clip,keepaspectratio]{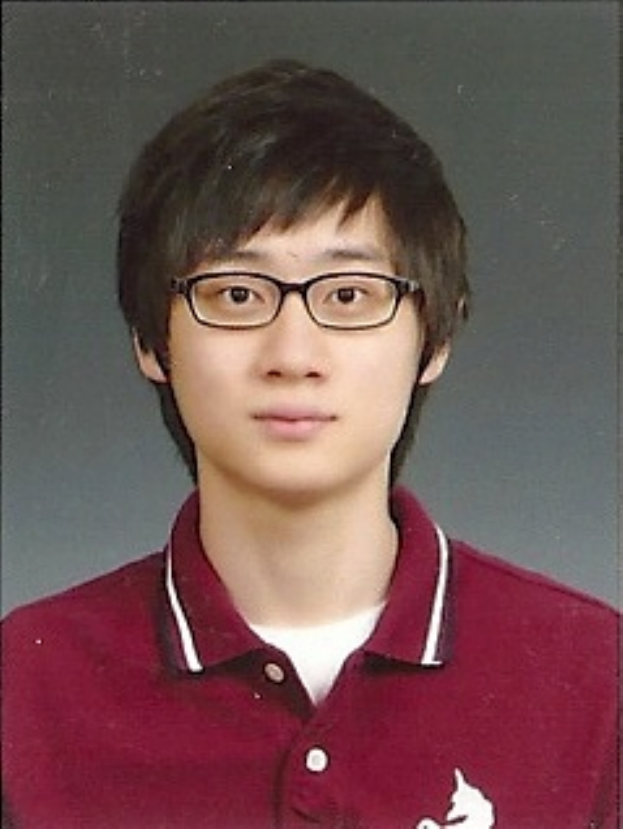}}]{Sung Whan Yoon}
	(S'12) received the B.S. and M.S. degrees in electrical engineering from the Korea Advanced Institute of Science and Technology (KAIST), Daejeon, Korea, in 2011 and 2013, respectively. He is currently pursuing the Ph.D degree in KAIST. His main research interests are in the field of coding and signal processing for wireless communication \& storage, especially massive MIMO, polar codes and distributed storage codes. He is a co-recipient of the IEEE international conference on
	communications (ICC) best paper award in 2017.
\end{IEEEbiography}
\begin{IEEEbiography}[{\includegraphics[width=1in,height=1.25in,clip,keepaspectratio]{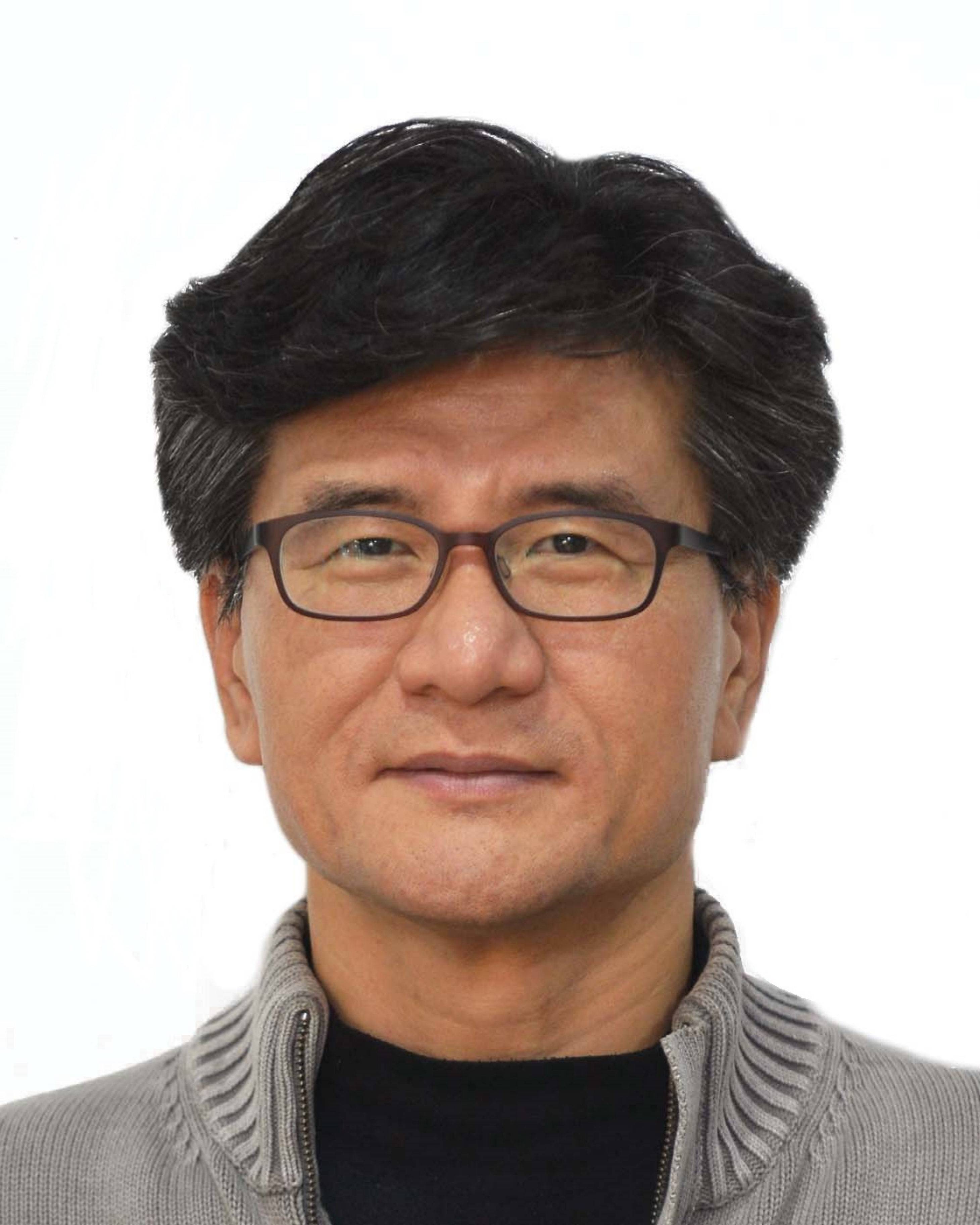}}]{Jaekyun Moon}
	(F'05) received the Ph.D degree in electrical and computer engineering at Carnegie Mellon University, Pittsburgh, Pa, USA. He is currently a Professor of electrical engineering at KAIST. From 1990 through early 2009, he was with the faculty of the Department of Electrical and Computer Engineering at the University of Minnesota, Twin Cities. He consulted as Chief Scientist for DSPG, Inc. from 2004 to 2007. He also worked as Chief Technology Officer at Link-A-Media Devices Corporation. His research interests are in the area of channel characterization, signal processing and coding for data storage and digital communication. Prof. Moon received the McKnight Land-Grant Professorship from the University of Minnesota. He received the IBM Faculty Development Awards as well as the IBM Partnership Awards. He was awarded the National Storage Industry Consortium (NSIC) Technical Achievement Award for the invention of the maximum transition run (MTR) code, a widely used error-control/modulation code in commercial storage systems. He served as Program Chair for the 1997 IEEE Magnetic Recording Conference. He is also Past Chair of the Signal Processing for Storage Technical Committee of the IEEE Communications Society, In 2001, he cofounded Bermai, Inc., a fabless semiconductor start-up, and served as founding President and CTO. He served as a guest editor for the 2001 IEEE JSAC issue on Signal Processing for High Density Recording. He also served as an Editor for IEEE TRANSACTIONS ON MAGNETICS in the area of signal processing and coding for 2001-2006. He is an IEEE Fellow.
\end{IEEEbiography}\vfill


%\newpage



%\vfill


\end{document}